\documentclass[11pt]{article}%
\usepackage{amsmath}
\usepackage{amsfonts}
\usepackage{amssymb}
\usepackage{amsthm}
\usepackage{graphicx}
\usepackage{tabularx}
\usepackage[table]{xcolor}
\usepackage{cite}
\usepackage{hyperref}
\usepackage{float}
\usepackage[font=small,labelfont=bf,justification=justified,singlelinecheck=false]%
{caption}
\usepackage{fullpage}%
\providecommand{\U}[1]{\protect\rule{.1in}{.1in}}
\definecolor{silver}{RGB}{200,200,200}
\newtheorem{theorem}{Theorem}

\newtheorem{claim}[theorem]{Claim}

\begin{document}

\title{Quantifying the Rise and Fall of Complexity in Closed Systems: The Coffee Automaton}
\author{Scott Aaronson\thanks{MIT. \ Email: aaronson@csail.mit.edu. \ \ This material
is based upon work supported by the National Science Foundation under Grant
No. 0844626, as well as an NSF Waterman Award.}
\and Sean M.\ Carroll\thanks{Walter Burke Institute for Theoretical Physics,
Caltech. \ Email: seancarroll@gmail.com. This research is funded in part by
DOE grant DE-FG02-92ER40701, and by the Gordon and Betty Moore Foundation
through Grant 776 to the Caltech Moore Center for Theoretical Cosmology and
Physics.}
\and Lauren Ouellette\thanks{This work was done while the author was a student at
MIT. \ Email: louelle@alum.mit.edu.}}
\date{}
\maketitle

\begin{abstract}
In contrast to entropy, which increases monotonically, the \textquotedblleft
complexity\textquotedblright\ or \textquotedblleft
interestingness\textquotedblright\ of closed systems seems intuitively to
increase at first and then decrease as equilibrium is approached. \ For
example, our universe lacked complex structures at the Big Bang and will also
lack them after black holes evaporate and particles are dispersed. \ This
paper makes an initial attempt to quantify this pattern. \ As a model system,
we use a simple, two-dimensional cellular automaton that simulates the mixing
of two liquids (\textquotedblleft coffee\textquotedblright\ and
\textquotedblleft cream\textquotedblright). \ A plausible complexity measure
is then the Kolmogorov complexity of a coarse-grained approximation of the
automaton's state, which we dub the \textquotedblleft apparent
complexity.\textquotedblright\ We study this complexity measure, and show
analytically that it never becomes large when the liquid particles are
non-interacting. \ By contrast, when the particles \textit{do} interact, we
give numerical evidence that the complexity reaches a maximum comparable to
the \textquotedblleft coffee cup's\textquotedblright\ horizontal dimension.
\ We raise the problem of proving this behavior analytically.

\end{abstract}

\hfill{CALT-68-2927}

\section{Introduction\label{INTRO}}

Imagine a cup of coffee into which cream has just been poured. \ At first, the
coffee and cream are separated. \ Over time, the two liquids diffuse until
they are completely mixed. \ If we consider the coffee cup a closed system, we
can say that its \textit{entropy} is increasing over time, in accordance with
the second law of thermodynamics. \ At the beginning, when the liquids are
completely separated, the system is in a highly ordered, low-entropy state.
\ After time has passed and the liquids have completely mixed, all of the
initial structure is lost; the system has high entropy.

Just as we can reason about the disorder of the coffee cup system, we can also
consider its \textquotedblleft complexity.\textquotedblright\ \ Informally, by
complexity we mean the amount of information needed to describe everything
\textquotedblleft interesting\textquotedblright\ about the system. \ At first,
when the cream has just been poured into the coffee, it is easy to describe
the state of the cup: it contains a layer of cream on top of a layer of
coffee. \ Similarly, it is easy to describe the state of the cup after the
liquids have mixed: it contains a uniform mixture of cream and coffee.
\ However, when the cup is in an intermediate state---where the liquids are
mixed in some areas but not in others---it seems more difficult to describe
what the contents of the cup look like.

Thus, it appears that the coffee cup system starts out at a state of low
complexity, and that the complexity first increases and then decreases over
time. \ In fact, this rising-falling pattern of complexity seems to hold true
for many closed systems. \ One example is the universe itself. \ The universe
began near the Big Bang in a low-entropy, low-complexity state, characterized
macroscopically as a smooth, hot, rapidly expanding plasma. It is predicted to
end in the high-entropy, low-complexity state of heat death, after black holes
have evaporated and the acceleration of the universe has dispersed all of the
particles (about $10^{100}$ years from now). \ But in between, complex
structures such as planets, stars, and galaxies have developed. \ There is no
general principle that quantifies and explains the existence of
high-complexity states at intermediate times in closed systems. \ It is the
aim of this work to explore such a principle, both by developing a more formal
definition of \textquotedblleft complexity,\textquotedblright\ and by running
numerical experiments to measure the complexity of a simulated coffee cup
system. \ The idea that complexity first increases and then decreases in as
entropy increases in closed system has been suggested informally
\cite{gell-mann,carroll}, but as far as we know this is the first quantitative
exploration of the phenomenon.

\section{Background\label{BG}}

Before discussing how to define \textquotedblleft
complexity,\textquotedblright\ let's start with the simpler question of how to
define entropy in a discrete dynamical system. \ There are various definitions
of entropy that are useful in different contexts. \ Physicists distinguish
between the Boltzmann and Gibbs entropies of physical systems. \ (There is
also the phenomenological thermodynamic entropy and the quantum-mechanical
von~Neumann entropy, neither of which are relevant here.) \ The Boltzmann
entropy is an objective feature of a microstate, but depends on a choice of
coarse-graining. \ We imagine coarse-graining the space of microstates into
equivalence classes, so that each microstate $x_{a}$ is an element of a unique
macrostate $X_{A}$. \ The volume $W_{A}$ of the macrostate is just the number
of associated microstates $x_{a}\in X_{A}$. \ Then the Boltzmann entropy of a
microstate $x_{a}$ is the normalized logarithm of the volume of the associated
macrostate:
\begin{equation}
S_{\mathrm{Boltzmann}}(x_{a}):=k_{B}\log W_{A},
\end{equation}
where $k_{B}$ is Boltzmann's constant (which we can set equal to $1$). \ The
Boltzmann entropy is independent of our knowledge of the system; in
particular, it can be nonzero even when we know the exact microstate. \ The
Gibbs entropy (which was also studied by Boltzmann), in contrast, refers to a
distribution function $\rho(x)$ over the space of microstates, which can be
thought of as characterizing our ignorance of the exact state of the system.
\ It is given by
\begin{equation}
S_{\mathrm{Gibbs}}[\rho]:=-\sum_{x}\rho(x)\log\rho(x).
\end{equation}
In probability theory, communications, information theory, and other areas,
the Shannon entropy of a probability distribution $D=(p_{x})_{x}$ is the
expected number of random bits needed to output a sample from the
distribution:
\begin{equation}
H(D):=-\sum_{x}p_x\log p_{x}.
\end{equation}
We see that this is essentially equivalent to the Gibbs entropy, with a slight
change of notation and vocabulary.

Finally, in computability theory, the entropy of an $n$-bit string $x$ is
often identified with its \emph{Kolmogorov complexity} $K(x)$: the length of
the shortest computer program that outputs $x$.\footnote{A crucial fact
justifying this definition is that switching from one (Turing-universal)
programming language to another changes $K(x)$ by at most an additive
constant, independent of $x$. \ The reason is that in one language, we can
always just write a compiler or interpreter for another language, then specify
$x$ using the second language. \ Also, throughout this paper, we will assume
for convenience that the program receives $x$'s length $n$ as input. \ This
assumption can change $K(x)$ by at most an additive $O(\log n)$ term.}
\ Strings that are highly patterned---meaning low in disorder---can be
described by a short program that takes advantage of those patterns. \ For
example, a string consisting of $n$ ones could be output by a short program
which simply loops $n$ times, printing `1' each time. \ Conversely, strings
which have little regularity cannot be compressed in this way. \ For such
strings, the shortest program to output them might simply be one that
hard-codes the entire string.

Fortunately, these notions of entropy are closely related to each other, so
that one can often switch between them depending on convenience. \ The Gibbs
and Shannon entropies are clearly equivalent. \ The Boltzmann entropy is
equivalent to the Gibbs entropy under the assumption that the distribution
function is flat over microstates within the given macrostate, and zero
elsewhere--\textit{i.e.}, given the knowledge of the system we would actually
obtain via macroscopic observation. \ For a computable distribution $D$ over
$n$-bit strings, the Kolmogorov complexity of a string sampled from $D$ tends
to the entropy of $D$ \cite{livitanyi}. \ (Thus, the Kolmogorov complexity of
a sequence of random numbers will be very high, even though there is no
\textquotedblleft interesting structure\textquotedblright in it.)

Despite these formal connections, the three kinds of entropy are calculated in
very different ways. \ The Boltzmann entropy is well-defined once a specific
coarse-graining is chosen. \ To estimate the Shannon entropy $H (D)$ of a
distribution $D$ (which we will henceforth treat as identical to the
corresponding Gibbs entropy), one in general requires knowledge of the entire
distribution $D$, which could potentially require exponentially many samples
from $D$. \ At first glance, the Kolmogorov complexity $K (x)$ seems even
worse: it is well-known to be \emph{uncomputable} (in fact, computing $K (x)$
is equivalent to solving the halting problem). \ On the other hand, in
practice one can often estimate $K (x)$ reasonably well by the compressed file
size, when $x$ is fed to a standard compression program such as \texttt{gzip}.
\ And crucially, unlike Shannon entropy, Kolmogorov complexity is well-defined
even for an individual string $x$. \ For these reasons, we chose to use $K
(x)$ (or rather, a computable approximation to it) as our estimate of entropy.

Of course, \textit{none} of the three measures of entropy capture
\textquotedblleft complexity,\textquotedblright\ in the sense discussed in
Section \ref{INTRO}. \ Boltzmann entropy, Shannon entropy, and Kolmogorov
complexity are all maximized by \textquotedblleft random\textquotedblright\ or
\textquotedblleft generic\textquotedblright\ objects and distributions,
whereas a complexity measure should be low both for \textquotedblleft
simple\textquotedblright\ objects \textit{and} for \textquotedblleft
random\textquotedblright\ objects, and large only for \textquotedblleft
interesting\textquotedblright\ objects that are neither simple nor random.

This issue has been extensively discussed in the complex systems and
algorithmic information theory communities since the 1980s. \ We are aware of
four quantitative ideas for how to define \textquotedblleft
complexity\textquotedblright\ or \textquotedblleft
interestingness\textquotedblright\ as distinct from entropy. \ While the
definitions look extremely different, it will turn out happily that they are
all related to one another, much like with the different definitions of
entropy. \ Note that our primary interest here is in the complexity of a
configuration defined at a single moment in time. One may also associate
measures of complexity to dynamical processes, which for the most part we
won't discuss.

\subsection{Apparent Complexity\label{APPCOMP}}

The first notion, and arguably the one that matches our intuition most
directly, we call \textit{apparent complexity}.\footnote{Here we are using
\textquotedblleft apparent\textquotedblright\ in the sense of
\textquotedblleft directly perceivable,\textquotedblright\ without meaning to
imply any connotation of \textquotedblleft illusory.\textquotedblright} \ By
the apparent complexity of an object $x$, we mean $H\left(  f\left(  x\right)
\right)  $, where $H$ is any of the entropy measures discussed previously, and
$f$ is some \textquotedblleft denoising\textquotedblright\ or
\textquotedblleft smoothing\textquotedblright\ function---that is, a function
that attempts to remove the \textquotedblleft incidental\textquotedblright\ or
\textquotedblleft random\textquotedblright\ information in $x$, leaving only
the \textquotedblleft interesting, non-random\textquotedblright\ information.
\ For example, if $x$\ is a bitmap image, then $f\left(  x\right)  $\ might
simply be a blurred version of $x$.

Apparent complexity has two immense advantages. \ First, it is simple: it
directly captures the intuition that we want something \textit{like} entropy,
but that leaves out \textquotedblleft incidental\textquotedblright%
\ information. \ For example, while the Kolmogorov complexity of a random
sequence would be very large, the apparent complexity of the same sequence
would typically be quite small, since the smoothing procedure would average
out the random fluctuations. \ Second, we can plausibly hope to
\textit{compute} (or at least, approximate) apparent complexity: we need
\textquotedblleft merely\textquotedblright\ solve\ the problems of computing
$H$ and $f$. \ It's because of these advantages that the complexity measure we
ultimately adopt for our experiments will be an approximate variant of
apparent complexity.

On the other hand, apparent complexity also has a large disadvantage: namely,
the apparent arbitrariness in the choice of the denoising function $f$. \ Who
decides which information about $x$ is \textquotedblleft
interesting,\textquotedblright\ and which is \textquotedblleft
incidental\textquotedblright? \ Won't $f$ depend, not only on the type of
object under study (bitmap images, audio recordings, etc.), but even more
worryingly, on the prejudices of the investigator? \ For example, suppose we
choose $f$\ to blur out details of an image that are barely noticeable to the
human eye. \ Then will studying the time-evolution of $H\left(  f\left(
x\right)  \right)  $ tell us anything about $x$ itself, or only about various
quirks of the human visual system?

Fortunately, the apparent arbitrariness of the smoothing procedure is less of
a problem than might initially be imagined. \ It is very much like the need
for a coarse-graining on phase space when one defines the Boltzmann entropy.
\ In either case, these apparently-arbitrary choices are in fact
well-motivated on physical grounds. \ While one \textit{could} choose bizarre
non-local ways to coarse-grain or smooth a distribution, natural choices are
typically suggested by our physical ability to actually observe systems, as
well as knowledge of their dynamical properties (see for example
\cite{brun1999classical}). \ When deriving the equations of fluid dynamics
from kinetic theory, in principle one could choose to average over cells of
momentum space rather than in position space; but there is no physical reason
to do so, since interactions are local in position rather than momentum.
\ Likewise, when we observe configurations (whether with our eyes, or with
telescopes or microscopes), large-scale features are more easily discerned
than small-scale ones. \ (In field theory this feature is formalized by the
renormalization group.) \ It therefore makes sense to smooth configurations
over local regions in space.

Nevertheless, we would ideally like our complexity measure to \textit{tell us}
what the distinction between \textquotedblleft random\textquotedblright\ and
\textquotedblleft non-random\textquotedblright\ information consists of,
rather than having to decide ourselves on a case-by-case basis. \ This
motivates an examination of some alternative complexity measures.

\subsection{Sophistication\label{SOPH}}

The second notion---one that originates in work of Kolmogorov himself---is
\textit{sophistication}. \ Roughly speaking, sophistication seeks to
generalize Kolmogorov complexity to capture only the non-random information in
a string---while using Kolmogorov complexity itself to define what is meant by
\textquotedblleft non-random.\textquotedblright\ \ Given an $n$-bit string
$x$, let a \textit{model} for $x$ be a set $S\subseteq\left\{  0,1\right\}
^{n}$\ such that $x\in S$. \ Let $K\left(  S\right)  $ be the length of the
shortest program that enumerates the elements of $S$, in any order (crucially,
the program must halt when it is done enumerating the elements). \ Also, let
$K\left(  x|S\right)  $\ be the length of the shortest program that outputs
$x$ given as input a description of $x$. \ Then we can consider $x$ to be a
\textquotedblleft generic\textquotedblright\ element of $S$ if $K\left(
x|S\right)  \geq\log_{2}\left\vert S\right\vert -c$ for some small constant
$c$. \ This means intuitively that $S$ is a \textquotedblleft
maximal\textquotedblright\ model for $x$: one can summarize all the
interesting, non-random properties of $x$\ by simply saying that $x\in S$.

Now the $c$-\textit{sophistication} of $x$ or $\operatorname*{soph}_{c}\left(
x\right)  $, defined by Koppel \cite{koppel}, is the minimum of $K\left(
S\right)  $\ over all models $S$ for $x$ such that $K\left(  S\right)
+\log_{2}\left\vert S\right\vert \leq K\left(  x\right)  +c$. \ (The optimal
such $S$ is said to \textquotedblleft witness\textquotedblright%
\ $\operatorname*{soph}_{c}\left(  x\right)  $.) \ In words,
$\operatorname*{soph}_{c}\left(  x\right)  $\ is the smallest possible amount
of \textquotedblleft non-random\textquotedblright\ information in a program
for $x$ that consists of two parts---a \textquotedblleft
non-random\textquotedblright\ part (specifying $S$)\ and a \textquotedblleft
random\textquotedblright\ part (specifying $x$ within $S$)---assuming the
program is also near-minimal. \ We observe the following:

\begin{itemize}
\item[(i)] $\operatorname*{soph}_{c}\left(  x\right)  \leq K\left(  x\right)
+O\left(  1\right)  $, since we can always just take $S=\left\{  x\right\}  $
as our model for $x$.

\item[(ii)] \textit{Most} strings $x$ satisfy $\operatorname*{soph}_{c}\left(
x\right)  =O\left(  1\right)  $, since we can take $S=\left\{  0,1\right\}
^{n}$\ as our model for $x$.

\item[(iii)] If $S$ witnesses $\operatorname*{soph}_{c}\left(  x\right)  $,
then $\log_{2}\left\vert S\right\vert \leq K\left(  x\right)  -K\left(
S\right)  +c\leq K\left(  x|S\right)  +c$, meaning that $x$\ must be a
\textquotedblleft generic\textquotedblright\ element of $S$.
\end{itemize}

It can be shown (see G\'{a}cs, Tromp, and Vit\'{a}nyi \cite{gtv}\ or
Antunes\ and Fortnow \cite{soph}) that there do exist highly \textquotedblleft
sophisticated\textquotedblright\ strings $x$, which satisfy
$\operatorname*{soph}_{c}\left(  x\right)  \geq n-c-O\left(  \log n\right)  $.
\ Interestingly, the proof of that result makes essential use of the
assumption that the program for $S$ halts, after it has finished listing $S$'s
elements. \ If we dropped that assumption, then we could always achieve
$K\left(  S\right)  =O\left(  \log n\right)  $, by simply taking $S$ to be the
set of all $y\in\left\{  0,1\right\}  ^{n}$ such that $K\left(  y\right)  \leq
K\left(  x\right)  $, and enumerating those $y$'s in a dovetailing fashion.

Recently, Mota et al.\ \cite{maas}\ studied a natural variant of
sophistication, in which one only demands that $S$\ be a maximal model for $x$
(i.e., that $K\left(  x|S\right)  \geq\log_{2}\left\vert S\right\vert -c$),
and not that $S$ also lead to a near-optimal two-part program for $x$ (i.e.,
that $K\left(  S\right)  +\log_{2}\left\vert S\right\vert \leq K\left(
x\right)  +c$). \ More formally, Mota et al.\ define the \textit{na\"{\i}ve
}$c$\textit{-sophistication} of $x$, or $\operatorname*{nsoph}_{c}\left(
x\right)  $, to be the minimum of $K\left(  S\right)  $\ over all models $S$
for $x$ such that $K\left(  x|S\right)  \geq\log_{2}\left\vert S\right\vert
-c$. \ By point (iii) above, it is clear that $\operatorname*{nsoph}%
_{c}\left(  x\right)  \leq\operatorname*{soph}_{c}\left(  x\right)  $.
\ \textit{A priori}, $\operatorname*{nsoph}_{c}\left(  x\right)  $\ could be
much smaller $\operatorname*{soph}_{c}\left(  x\right)  $, thereby leading to
two different sophistication notions. \ However, it follows from an important
2004 result of Vereshchagin and Vit\'{a}nyi \cite{vereshvitanyi} that
$\operatorname*{soph}_{c+O\left(  \log n\right)  }\left(  x\right)
\leq\operatorname*{nsoph}_{c}\left(  x\right)  $\ for all $x$, and hence the
two notions are basically equivalent.

Sophistication is sometimes criticized for being \textquotedblleft
brittle\textquotedblright: it is known that increasing the parameter $c$\ only
slightly can cause $\operatorname*{soph}_{c}\left(  x\right)  $\ and
$\operatorname*{nsoph}_{c}\left(  x\right)  $\ to fall drastically, say from
$n-O\left(  \log n\right)  $\ to $O\left(  1\right)  $. \ However, a simple
fix to that problem is to consider the quantities $\min_{c}\left\{
c+\operatorname*{soph}_{c}\left(  x\right)  \right\}  $\ and $\min_{c}\left\{
c+\operatorname*{nsoph}_{c}\left(  x\right)  \right\}  $. \ Those are known,
respectively, as the \textit{coarse sophistication} $\operatorname*{csoph}%
\left(  x\right)  $\ and \textit{na\"{\i}ve coarse sophistication}
$\operatorname*{ncsoph}\left(  x\right)  $, and they satisfy
$\operatorname*{ncsoph}\left(  x\right)  \leq\operatorname*{csoph}\left(
x\right)  \leq\operatorname*{ncsoph}\left(  x\right)  +O\left(  \log n\right)
$.

The advantage of sophistication is that it captures, more cleanly than any
other measure, what exactly we mean by \textquotedblleft
interesting\textquotedblright\ versus \textquotedblleft
random\textquotedblright\ information. \ Unlike with apparent complexity, with
sophistication there's no need to specify a smoothing function $f$, with the
arbitrariness that seems to entail. \ Instead, if one likes, the definition of
sophistication picks out a smoothing function for us: namely, whatever
function maps $x$\ to its corresponding model $S$.

Unfortunately, this conceptual benefit\ comes at a huge computational price.
\ Just as $K\left(  x\right)  $\ is uncomputable, so one can show that the
sophistication measures are uncomputable as well. \ But with $K\left(
x\right)  $, at least we can get better and better upper bounds, by finding
smaller and smaller compressed representations for $x$. \ By contrast, even to
\textit{approximate} sophistication requires solving two coupled optimization
problems: firstly over possible models $S$, and secondly over possible ways to
specify $x$ given $S$.

A second disadvantage of sophistication is that, while there \textit{are}
highly-sophisticated strings, the only known way to produce such a string
(even probabilistically) is via a somewhat-exotic diagonalization argument.
\ (By contrast, for \textquotedblleft reasonable\textquotedblright\ choices of
smoothing function $f$, one can easily generate $x$ for which the apparent
complexity $H\left(  f\left(  x\right)  \right)  $\ is large.) \ Furthermore,
this is not an accident, but an unavoidable consequence of sophistication's
generality. \ To see this, consider any short probabilistic program $P$: for
example, the coffee automaton that we will study in this paper, which has a
simple initial state and a simple probabilistic evolution rule. \ Then we
claim that \textit{with overwhelming probability, }$P$\textit{'s output }%
$x$\textit{\ must have low sophistication.} \ For as the model $S$, one can
take the set of all possible outputs $y$\ of $P$ such that $\Pr\left[
y\right]  \approx\Pr\left[  x\right]  $. \ This $S$\ takes only $O\left(  \log
n\right)  $\ bits to describe (plus $O\left(  1\right)  $\ bits for $P$
itself), and clearly $K\left(  x|S\right)  \geq\log_{2}\left\vert S\right\vert
-c$\ with high probability over $x$.

For this reason, sophistication as defined above seems irrelevant to the
coffee cup or other physical systems:\ it simply never becomes large for such
systems! \ On the other hand, note that the two drawbacks of sophistication
might \textquotedblleft cancel each other out\textquotedblright\ if we
consider \textit{resource-bounded} versions of sophistication: that is,
versions where we impose constraints (possibly severe constraints) on both the
program for generating $S$, and the program for generating $x$ given $S$.
\ Not only does the above argument fail for resource-bounded versions of
sophistication, but those versions are the only ones we can hope to compute
anyway! \ With Kolmogorov complexity, we're forced to consider proxies (such
as \texttt{gzip} file size) mostly just because $K\left(  x\right)  $\ itself
is uncomputable. \ By contrast, even if we \textit{could} compute
$\operatorname*{soph}_{c}\left(  x\right)  $\ perfectly, it would never become
large for the systems that interest us here.

\subsection{Logical Depth\label{DEPTH}}

A third notion, introduced by Bennett \cite{bennett:depth}, is \textit{logical
depth}. \ Roughly speaking, the logical depth of a string $x$ is the amount of
time taken by the shortest program that outputs $x$. \ (Actually, to avoid the
problem of \textquotedblleft brittleness,\textquotedblright\ one typically
considers something like the minimum amount of time taken by any program that
outputs $x$ and whose length is at most $K\left(  x\right)  +c$, for some
constant \textquotedblleft fudge factor\textquotedblright\ $c$. \ This is
closely analogous to what is done for sophistication.)

The basic idea here is that, both for simple strings and for random ones, the
shortest program will \textit{also} probably run in nearly linear time. \ By
contrast, one can show that there exist \textquotedblleft
deep\textquotedblright\ strings, which can be generated by short programs but
only after large amounts of time.

Like sophistication, logical depth tries to probe the internal structure of a
minimal program for $x$---and in particular, to distinguish between the
\textquotedblleft interesting code\textquotedblright\ in that program\ and the
\textquotedblleft boring data\textquotedblright\ on which the code acts. \ The
difference is that, rather than trying to measure the \textit{size} of the
\textquotedblleft interesting code,\textquotedblright\ one examines how long
it takes to run.

Bennett \cite{bennett:depth} has advocated logical depth as a complexity
measure, on the grounds that logical depth encodes the \textquotedblleft
amount of computational effort\textquotedblright\ used to produce $x$,
according to the \textquotedblleft most probable\textquotedblright\ (i.e.,
lowest Kolmogorov complexity) hypothesis about how $x$ was generated. \ On the
other hand, an obvious disadvantage of logical depth is that it's even less
clear how to estimate it in practice than was the case for sophistication.

A second objection to logical depth is that even short, fast programs can be
extremely \textquotedblleft complicated\textquotedblright\ in their behavior
(as evidenced, for example, by cellular automata such as Conway's Game of
Life). \ Generating what many people would regard as a visually complex
pattern---and what \textit{we} would regard as a complex, milk-tendril-filled
state, in the coffee-cup system---simply need not take a long time! \ For this
reason, one might be uneasy with the use of running time as a proxy for complexity.

\subsection{Light-Cone Complexity\label{LCC}}

The final complexity measure we consider was proposed by Shalizi, Shalizi, and
Haslinger \cite{ssh}; we call it \textit{light-cone complexity}. \ In contrast
to the previous measures, light-cone complexity does not even try to define
the \textquotedblleft complexity\textquotedblright\ of a string $x$, given
only $x$ itself. \ Instead, the definition of light-cone complexity assumes a
\textit{causal structure}: that is, a collection of spacetime points
$\mathcal{A}$ (assumed to be fixed), together with a transitive, cycle-free
binary relation indicating which points $a\in\mathcal{A}$\ are to the
\textquotedblleft future\textquotedblright\ of which other points in
$\mathcal{A}$. \ The set of all points $b\in\mathcal{A}$\ to $a$'s future is
called $a$'s \textit{future light-cone}, and is denoted $F\left(  a\right)  $.
\ The set of all points $b\in\mathcal{A}$\ to $a$'s past (that is, such that
$a$ is to $b$'s future)\ is called $a$'s \textit{past light-cone}, and is
denoted $P\left(  a\right)  $. \ For example, if we were studying the
evolution of a $1$-dimensional cellular automaton, then $\mathcal{A}$\ would
consist of all ordered pairs $\left(  x,t\right)  $\ (where $x$ is position
and $t$ is time), and we would have%
\begin{align}
F\left(  x,t\right)   &  =\left\{  \left(  y,u\right)  :u>t,~~\left\vert
x-y\right\vert \leq u-t\right\}  ,\\
P\left(  x,t\right)   &  =\left\{  \left(  y,u\right)  :u<t,~~\left\vert
x-y\right\vert \leq t-u\right\}  .
\end{align}

Now given a spacetime point $a\in\mathcal{A}$, let $V_{a}$\ be the actual
value assumed by the finite automaton at $a$\ (for example, \textquotedblleft
alive\textquotedblright\ or \textquotedblleft dead,\textquotedblright\ were we
discussing Conway's Game of Life or some other $2$-state system). \ In
general, the finite automaton might be probabilistic, in which case $V_{a}%
$\ is a random variable, with a Shannon entropy $H\left(  V_{a}\right)  $\ and
so forth. \ Also, given a set $S\subseteq\mathcal{A}$, let $V_{S}:=\left(
V_{a}\right)  _{a\in S}$ be a complete description of the values at
\textit{all} points in $S$. \ Then the light-cone complexity at a point
$a\in\mathcal{A}$, or $\operatorname*{LCC}\left(  a\right)  $, can be defined
as follows:%
\begin{align}
\operatorname*{LCC}\left(  a\right)   &  =I\left(  V_{P\left(  a\right)
}:V_{F\left(  a\right)  }\right)  \\
&  =H\left(  V_{P\left(  a\right)  }\right)  +H\left(  V_{F\left(  a\right)
}\right)  -H\left(  V_{P\left(  a\right)  },V_{F\left(  a\right)  }\right)  .
\end{align}
In other words, $\operatorname*{LCC}\left(  a\right)  $\ is the \textit{mutual
information} between $a$'s past and future light-cones: the number of bits
about $a$'s future that are encoded by its past. \ If we want the light-cone
complexity of (say) an entire spatial slice, we could then take the sum of
$\operatorname*{LCC}\left(  a\right)  $\ over all $a$ in that slice, or some
other combination.

The intuition here is that, if the cellular automaton dynamics are
\textquotedblleft too simple,\textquotedblright\ then $\operatorname*{LCC}%
\left(  a\right)  $\ will be small simply because $H\left(  V_{P\left(
a\right)  }\right)  $\ and $H\left(  V_{F\left(  a\right)  }\right)  $\ are
both small. \ Conversely, if the dynamics are \textquotedblleft too
random,\textquotedblright\ then $\operatorname*{LCC}\left(  a\right)  $\ will
be small because $H\left(  V_{P\left(  a\right)  },V_{F\left(  a\right)
}\right)  \approx H\left(  V_{P\left(  a\right)  }\right)  +H\left(
V_{F\left(  a\right)  }\right)  $: although the past and future light-cones
both have plenty of entropy, they are uncorrelated, so that knowledge of the
past is of barely any use in predicting the future. \ Only in an intermediate
regime, where there are interesting \textit{non}-random dynamics, should there
be substantial uncertainty about $V_{F\left(  a\right)  }$\ that can be
reduced by knowing $V_{P\left(  a\right)  }$.

As Shalizi et al.\ \cite{ssh}\ point out, a major advantage of light-cone
complexity, compared to sophistication, logical depth, and so on, is that
light-cone complexity has a clear \textquotedblleft operational
meaning\textquotedblright: it is easy to state the question that light-cone
complexity is answering. \ That question is the following: \textquotedblleft
how much could I possibly predict about the configurations in $a$'s future,
given complete information about $a$'s past?\textquotedblright\ \ The reason
to focus on light-cones, rather than other sets of points, is that the
light-cones are automatically determined once we know the causal structure:
there seems to be little arbitrariness about them.

On the other hand, depending on the application, an obvious drawback of
light-cone complexity is that it can't tell us the \textquotedblleft
inherent\textquotedblright\ complexity of an object $x$, without knowing about
$x$'s past and future. \ If we wanted to use a complexity measure to make
\textit{inferences} about $x$'s past and future, this might be seen as
question-begging. \ A less obvious drawback arises if we consider a dynamical
system that changes slowly with time: for example, a version of the coffee
automaton where just a single cream particle is randomly moved at each time
step. \ Consider such a system in its \textquotedblleft late\textquotedblright%
\ stages: that is, after the coffee and cream have fully mixed. \ Even then,
Shalizi et al.'s $\operatorname*{LCC}\left(  a\right)  $\ measure will remain
large, but not for any \textquotedblleft interesting\textquotedblright%
\ reason: only because $a$'s past light-cone will contain almost the same
(random) information as its future light-cone, out to a very large distance!
\ Thus, $\operatorname*{LCC}$\ seems to give an intuitively wrong answer in
these cases (though no doubt one could address the problem by redefining
$\operatorname*{LCC}$\ in some suitable way).

The computational situation for $\operatorname*{LCC}$\ seems neither better
nor worse to us than that for (say) apparent complexity or resource-bounded
sophistication. \ Since the light-cones $P\left(  a\right)  $\ and $V\left(
a\right)  $\ are formally infinite, a first step in estimating
$\operatorname*{LCC}\left(  a\right)  $---as Shalizi et al.\ point out---is to
impose some finite cutoff $t$ on the number of steps into $a$'s past and
future one is willing to look. \ Even then, one needs to estimate the mutual
information $I\left(  V_{P_{t}\left(  a\right)  }:V_{F_{t}\left(  a\right)
}\right)  $\ between the truncated light-cones $P_{t}\left(  a\right)  $\ and
$F_{t}\left(  a\right)  $, a problem that na\"{\i}vely requires a number of
samples exponential in $t$. \ One could address this problem by simply taking
$t$ extremely small (Shalizi et al.\ set $t=1$). \ Alternatively, if a large
$t$ was needed, one could use the same Kolmogorov-complexity-based approach
that we adopt in this paper for apparent complexity. \ That is, one first
replaces the mutual information by the mutual \textit{algorithmic} information%
\begin{equation}
K\left(  V_{P_{t}\left(  a\right)  }:V_{F_{t}\left(  a\right)  }\right)
=K\left(  V_{P_{t}\left(  a\right)  }\right)  +K\left(  V_{F_{t}\left(
a\right)  }\right)  -K\left(  V_{P_{t}\left(  a\right)  },V_{F_{t}\left(
a\right)  }\right)  ,
\end{equation}
and then estimates $K\left(  x\right)  $\ using some computable proxy such as
\texttt{gzip}\ file size.

\subsection{Synthesis\label{SYNTHESIS}}

It seems like we have a bestiary of different complexity notions.
\ Fortunately, the four notions discussed above can all be related to each
other; let us discuss how.

First, one can view apparent complexity as a kind of \textquotedblleft
resource-bounded\textquotedblright\ sophistication. \ To see this, let $f$ be
any smoothing function. \ Then $K\left(  f\left(  x\right)  \right)  $, the
Kolmogorov complexity of $f\left(  x\right)  $, is essentially equal to
$K\left(  S_{f,x}\right)  $, where%
\begin{equation}
S_{f,x}:=\left\{  y:f\left(  y\right)  =f\left(  x\right)  \right\}
.\label{sfx}%
\end{equation}
Thus, if instead of minimizing over \textit{all} models $S$ for $x$ that
satisfy some condition, we consider only the particular model $S_{f,x}$ above,
then sophistication reduces to apparent complexity. \ Note that this argument
establishes neither that apparent complexity is an upper bound on
sophistication, nor that it's a lower bound. \ Apparent complexity could be
larger, if the minimization found some model $S$\ for $x$ with $K\left(
S\right)  \ll K\left(  S_{f,x}\right)  $. \ But conversely, sophistication
could also be larger, if the model $S_{f,x}$\ happened to satisfy $K\left(
x|S_{f,x}\right)  \ll\log_{2}\left\vert S_{f,x}\right\vert $ (that is, $x$ was
a highly \textquotedblleft non-generic\textquotedblright\ element of $S_{f,x}$).

Second, Antunes\ and Fortnow \cite{soph} proved a close relation between
coarse sophistication and a version of logical depth. \ Specifically, the Busy
Beaver function, $\operatorname*{BB}\left(  k\right)  $, is defined as the
maximum number of steps for which a $k$-bit program can run before halting
when given a blank input. \ Then given a string $x$, Antunes and Fortnow
\cite{soph}\ define the \textit{Busy Beaver computational depth}
$\operatorname*{depth}_{BB}\left(  x\right)  $ to be the minimum, over all
programs $p$ that output a model $S$ for $x$ in $BB\left(  k\right)  $\ steps
or fewer, of $\left\vert p\right\vert +k-K\left(  x\right)  $. They then prove
the striking result that $\operatorname*{csoph}$ and $\operatorname*{depth}%
_{BB}$ are essentially equivalent: for all $x\in\left\{  0,1\right\}  ^{n}$,%
\begin{equation}
\left\vert \operatorname*{csoph}\left(  x\right)  -\operatorname*{depth}%
\nolimits_{BB}\left(  x\right)  \right\vert =O\left(  \log n\right)  .
\end{equation}

Third, while light-cone complexity is rather different from the other three
measures (due to its taking as input an entire causal history), it can be
loosely related to apparent complexity as follows. \ If $\operatorname*{LCC}%
\left(  a\right)  $\ is large, then the region around $a$ must contain large
\textquotedblleft contingent structures\textquotedblright: structures that are
useful for predicting future evolution, but that might have been different in
a different run of the automaton. \ And one might expect those structures to
lead to a large apparent complexity in $a$'s vicinity. \ Conversely, if the
apparent complexity is large, then one expects contingent structures (such as
milk tendrils, in the coffee automaton), which could then lead to nontrivial
mutual information between $a$'s past and future light-cones.

Having described four complexity measures, their advantages and disadvantages,
and their relationships to each other, we now face the question of which
measure to use for our experiment. \ While it would be interesting to study
the rise and fall of light-cone complexity in future work, here we decided to
restrict ourselves to complexity measures that are functions of the current
state. \ That leaves apparent complexity, sophistication, and logical depth
(and various approximations, resource-bounded versions, and hybrids thereof).

Ultimately, we decided on a type of apparent complexity. \ Our reason was
simple: because even after allowing resource bounds, \textit{we did not know
of any efficient way to approximate sophistication or logical depth}. \ In
more detail, given a bitmap image $x$ of a coffee cup, our approach first
\textquotedblleft smears $x$ out\textquotedblright\ using a smoothing function
$f$, then uses the \texttt{gzip}\ file size of $f\left(  x\right)  $\ as an
upper bound on the Kolmogorov complexity $K\left(  f\left(  x\right)  \right)
$\ (which, in turn, is a proxy for the Shannon entropy $H\left(  f\left(
x\right)  \right)  $\ of $f\left(  x\right)  $\ considered as a random
variable). \ There are a few technical problems that arise when implementing
this approach (notably, the problem of \textquotedblleft border pixel
artifacts\textquotedblright). \ We discuss those problems and our solutions to
them in Section \ref{ALGS}.

Happily, as discussed earlier in this section, our apparent complexity measure
can be related to the other measures. \ For example, apparent complexity can
be seen as an \textit{extremely} resource-bounded variant of sophistication,
with the set $S_{f,x}$\ of equation (\ref{sfx}) playing the role of the model
$S$. \ As discussed in Section \ref{APPCOMP}, one might object to our apparent
complexity measure on the grounds that our smoothing function $f$ is
\textquotedblleft arbitrary,\textquotedblright\ that we had no principled
reason to choose it rather than some other function. \ Interestingly, though,
one can answer that objection by taking inspiration from light-cone
complexity. \ Our smoothing function $f$ will \textit{not} be completely
arbitrary, for the simple reason that the regions over which we
coarse-grain---namely, squares of contiguous cells---will correspond to the
coffee automaton's causal structure.\footnote{Technically, if we wanted to
follow the causal structure, then we should have used \textit{diamonds} of
continguous cells rather than squares. \ But this difference is presumably
insignificant.}

\section{The Coffee Automaton\label{SETUP}}

The coffee cup system that we use as our model is a simple stochastic cellular
automaton. \ A two-dimensional array of bits describes the system's state,
with ones representing particles of cream, and zeros representing particles of
coffee. \ The cellular automaton implementation used for this project is
written in Python; source code is available for download.\footnote{At
\texttt{www.scottaaronson.com/coffee\_automaton.zip}}

The automaton begins in a state in which the top half of the cells are filled
with ones, and the bottom half is filled with zeros. \ At each time step, the
values in the cells change according to a particular transition rule. \ We
consider two different models of the coffee cup system, each having its own
transition rule.

\subsection{Interacting Model\label{INT}}

In the interacting model of the coffee cup system, only one particle may
occupy each cell in the state array. \ The transition rule for this model is
as follows: at each time step, one pair of horizontally or vertically
adjacent, differing particles is selected, and the particles' positions are
swapped. \ This model is \textit{interacting} in the sense that the presence
of a particle in a cell prevents another particle from entering that cell.
\ The movements of particles in this model are not independent of one another.

This model reflects the physical principle that two pieces of matter may not
occupy the same space at the same time. \ However, the interactions between
particles that make this model more realistic also make it harder to reason
about theoretically.

\subsection{Non-Interacting Model\label{NONINT}}

In the non-interacting model of the coffee cup system, any number of cream
particles may occupy a single cell in the state array. \ Coffee particles are
not considered important in this model; they are simply considered a
background through which the cream particles move. \ The transition rule for
this model is as follows: at each time step, each cream particle in the system
moves one step in a randomly chosen direction. \ This model is
\textit{non-interacting} in that the location of each cream particle is
independent of all the others. \ The presence of a cream particle in a
particular cell does not prevent another cream particle from also moving into
that cell.

We consider this model because it is easier to understand theoretically.
\ Since the particles in the system do not interact, each particle can be
considered to be taking an independent random walk. \ The dynamics of random
walks are well-understood, so it is easy to make theoretical predictions about
this model (see Appendix \ref{ANALYSIS}) and compare them to the experimental results.

\section{Approximating Apparent Complexity\label{ALGS}}

While Kolmogorov complexity and sophistication are useful theoretical notions
to model our ideas of entropy and complexity, they cannot be directly applied
in numerical simulations, because they are both uncomputable. \ As such, while
we use these concepts as a theoretical foundation, we need to develop
algorithms that attempt to approximate them.

Evans et al.\ \cite{oscr} propose an algorithm, called the optimal symbol
compression ratio (OSCR) algorithm, which directly estimates Kolmogorov
complexity and sophistication. \ Given an input string $x$, the OSCR algorithm
produces a two-part code. \ The first part is a codebook, which maps symbols
chosen from the original input string to new symbols in the encoded string.
\ The second part of the code is the input string, encoded using the symbols
in this codebook. \ The goal of OSCR is to select which symbols to put in the
codebook such that the total size of the output---codebook size plus encoded
string size---is minimized. \ The optimal codebook size for $x$ is an estimate
of $K\left(  S\right)  $, the sophistication of $x$. \ The optimal total size
of the output for $x$ is called the minimum description length (MDL) of the
string, and is an estimate of $K\left(  x\right)  $.

The OSCR approach seems promising because of its direct relationship to the
functions we are interested in approximating. \ However, we implemented a
version of this algorithm, and we found that our implementation does not
perform well in compressing the automaton data. \ The output of the algorithm
is noisy, and there is no obvious trend in either the entropy or complexity
estimates. \ We conjecture that the noise is present because this compression
method, unlike others we consider, does not take into account the
two-dimensionality of the automaton state.

An alternative metric adopts the idea of coarse-graining. \ Here we aim to
describe a system's state on a macroscopic scale---for example, the coffee cup
as it would be seen by a human observer from a few feet away---by smoothing
the state, averaging nearby values together. \ Conceptually, for an automaton
state represented by a string $x$, its coarse-grained version is analogous to
a typical set $S$ which contains $x$. \ The coarse-grained state describes the
high-level, \textquotedblleft non-random\textquotedblright\ features of
$x$---features which it has in common with all other states from which the
same coarse-grained representation could be derived. \ Thus, the descriptive
size of the coarse-grained state can be used as an estimate for the state's
sophistication, $K\left(  S\right)  $. \ To estimate the descriptive size of
the coarse-grained state, we compress it using a general file compression
program, such as \texttt{gzip} or \texttt{bzip}. \ Shalizi \cite{shalizi}
objects to the use of such compression programs, claiming that they do not
provide consistently accurate entropy estimates and that they are too slow.
\ In our experiments, we have not seen either of these problems; our
simulations run in a reasonable amount of time and produce quite consistent
entropy estimates (see, for instance, Figure \ref{multicompress}). \ We
therefore use such compression programs throughout, though we consider
alternative approaches in Section \ref{CONC}.

Having defined the notion of coarse-graining, we can then define a two-part
code based on it. \ If the first part of the code---the typical set---is the
coarse-grained state, then the second part is $K\left(  x|S\right)  $, the
information needed to reconstruct the fine-grained state given the
coarse-grained version. \ The total compressed size of both parts of the code
is an estimate of the Kolmogorov complexity of the state, $K\left(  x\right)
$.

We attempted to implement such a two-part code, in which the second part was a
diff between the fine-grained and coarse-grained states. \ The fine-grained
state, $x$, could be uniquely reconstructed from the coarse-grained array and
the diff. \ In our implementation of this two-part code, our estimate of
$K\left(  x|S\right)  $ suffered from artifacts due to the way the diff was
represented. \ However, defining a two-part code based on coarse-graining is
possible in general.

In light of the artifacts produced by our implementation of the two-part code,
we chose to pursue a more direct approach using coarse-graining. \ We
continued to use the compressed size of the coarse-grained state as an
approximation of $K\left(  S\right)  $. \ However, instead of approximating
$K\left(  x|S\right)  $ and using $K\left(  S\right)  +K\left(  x|S\right)  $
as an estimate of $K\left(  x\right)  $, we approximated $K\left(  x\right)  $
directly, by measuring the compressed size of the fine-grained array. \ This
approach avoided the artifacts of the diff-based code, and was used to
generate the results reported here.

\section{Coarse-Graining Experiment\label{CG}}

\subsection{Method\label{CGMETHOD}}

To derive a coarse-grained version of the automaton state from its original,
fine-grained version, we construct a new array in which the value of each cell
is the average of the values of the nearby cells in the fine-grained array.
\ We define \textquotedblleft nearby\textquotedblright\ cells as those within
a $g\times g$ square centered at the cell in question. The value of $g$ is
called the grain size, and here is selected experimentally. \ This procedure
is illustrated in Figure \ref{finecoarsefig}.

\begin{figure}\begin{center}%
\begin{tabular}
[c]{cc}%
\textbf{Fine-Grained} & \textbf{Coarse-Grained}\\%
\begin{tabular}
[c]{|c|c|c|c|c|}\hline
0 & 1 & 0 & 0 & 1\\\hline
0 & \cellcolor{silver}1 & \cellcolor{silver}1 & \cellcolor{silver}0 &
1\\\hline
0 & \cellcolor{silver}1 & \cellcolor{silver}0 & \cellcolor{silver}1 &
0\\\hline
1 & \cellcolor{silver}0 & \cellcolor{silver}1 & \cellcolor{silver}1 &
1\\\hline
1 & 0 & 1 & 0 & 0\\\hline
\end{tabular}
&
\begin{tabular}
[c]{|c|c|c|c|c|}\hline
0.5 & 0.5 & 0.5 & 0.5 & 0.5\\\hline
0.5 & $0.\overline{4}$ & $0.\overline{4}$ & $0.\overline{4}$ & 0.5\\\hline
0.5 & $0.\overline{5}$ & \cellcolor{silver}$0.\overline{6}$ & $0.\overline{6}$
& $0.\overline{3}$\\\hline
0.5 & $0.\overline{5}$ & $0.\overline{5}$ & $0.\overline{5}$ & 0.5\\\hline
0.5 & $0.\overline{6}$ & 0.5 & $0.\overline{6}$ & 0.5\\\hline
\end{tabular}
\\
&
\end{tabular}
\end{center}
\caption{Illustration of the construction of the coarse-grained array, using
an example grain size of $3$. \ The values of the shaded cells at left are
averaged to produce the value of the shaded cell at right.}%
\label{finecoarsefig}%
\end{figure}

Given this array of averages, we then threshold its floating-point values into
three buckets. \ Visually, these buckets represent areas which are mostly
coffee (values close to $0$), mostly cream (values close to $1$), or mixed
(values close to $0.5$). \ The estimated complexity of the state, $K\left(
S\right)  $, is the file size of the thresholded, coarse-grained array after
compression. \ Analogously, the estimated entropy of the automaton state is
the compressed file size of the fine-grained array.

\subsection{Results and Analysis\label{CGRESULTS}}

Results from simulation of the automaton using the coarse-graining metric are
shown in Figure \ref{initialgraph}.

\begin{figure}[h]
\begin{center}
\fbox{\includegraphics[height=150px, width=200px]{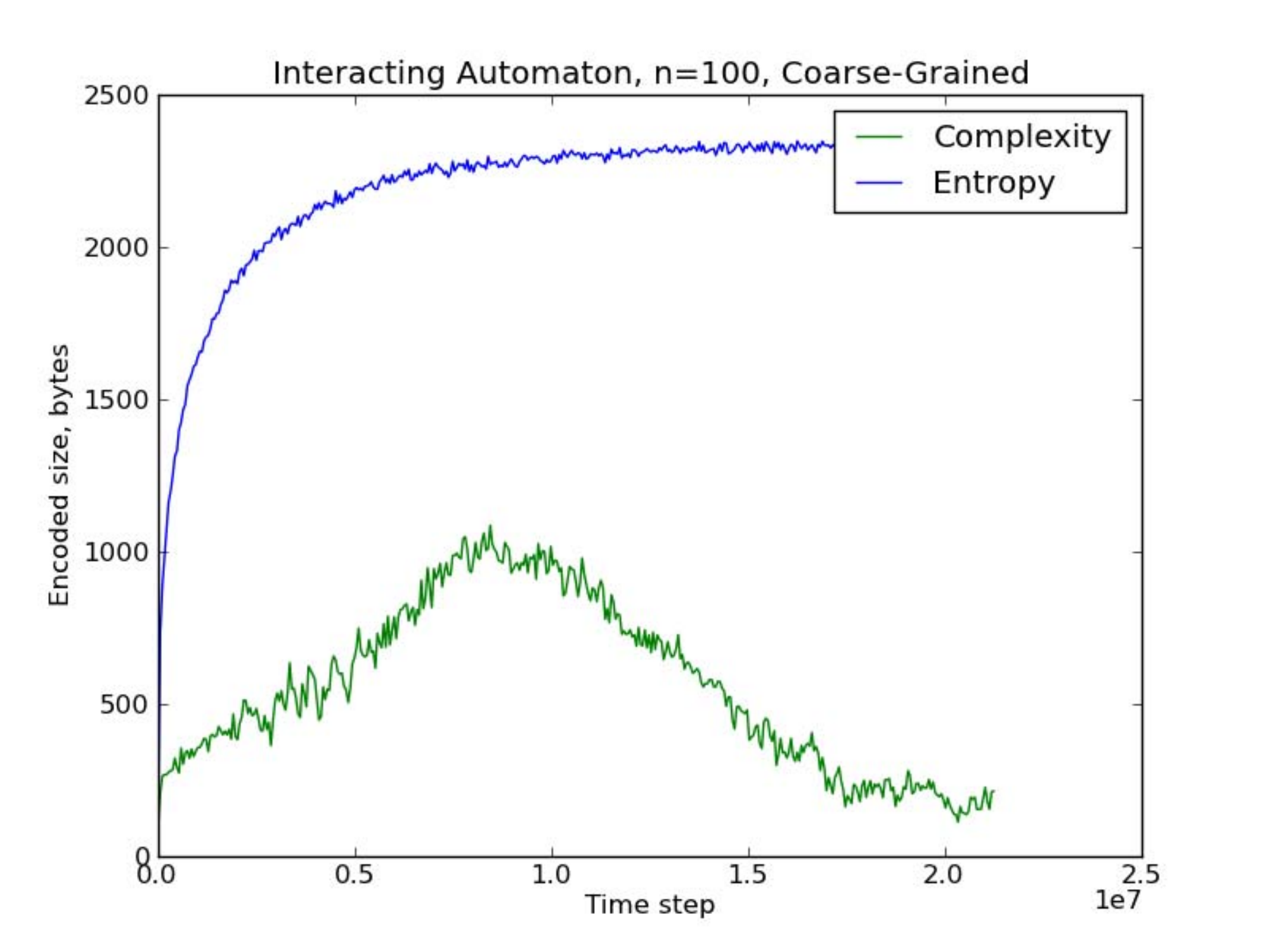}}
\fbox{\includegraphics[height=150px, width=200px]{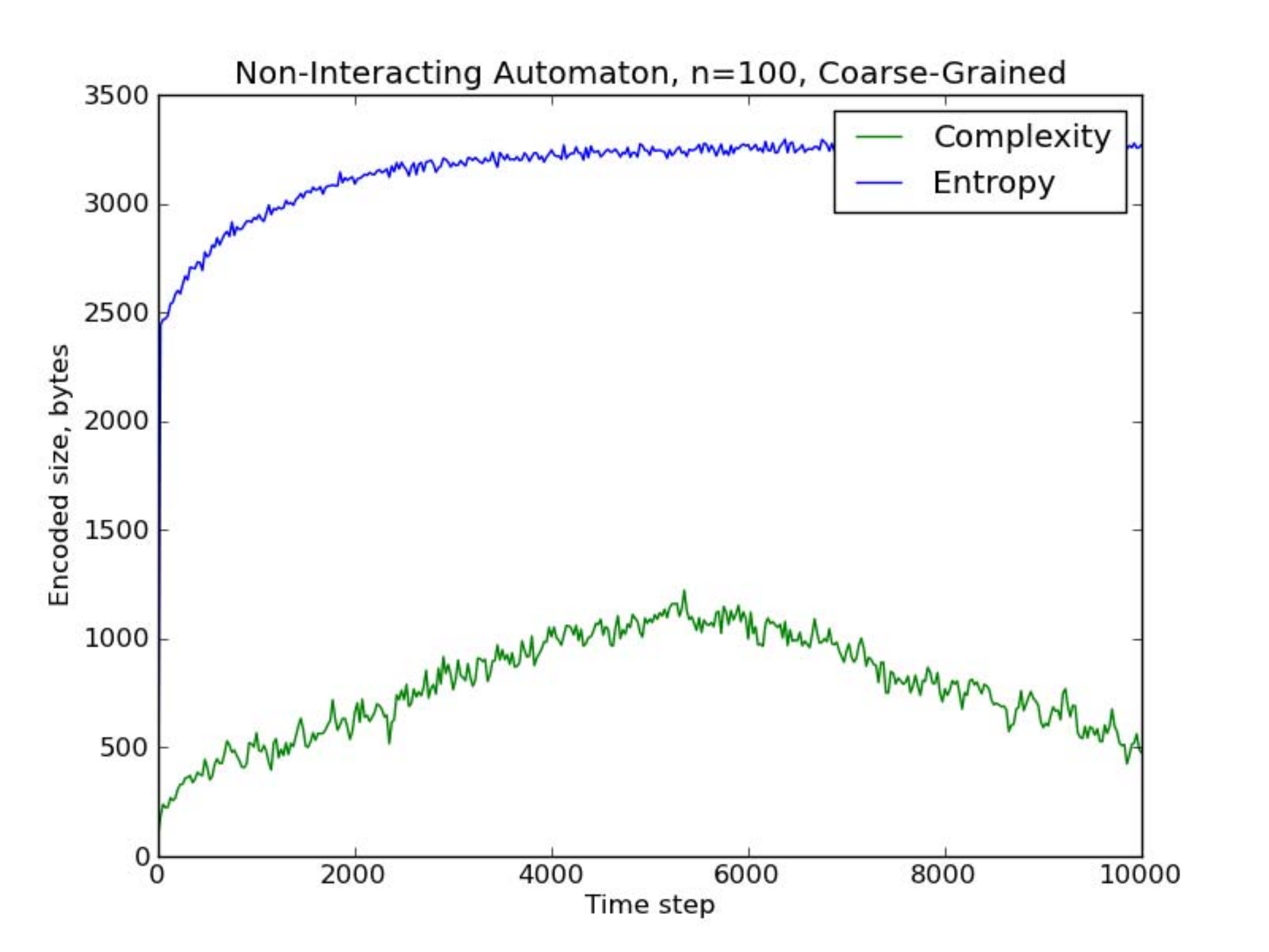}}%
\caption{The estimated entropy and complexity of an automaton using the
coarse-graining metric. Results for the interacting model are shown at left,
and results for the non-interacting model are at right.}%
\label{initialgraph}%
\end{center}
\end{figure}

Both the interacting and non-interacting models show the predicted increasing,
then decreasing pattern of complexity. \ Both models also have an increasing
entropy pattern, which is expected due to the second law of thermodynamics.
\ The initial spike in entropy for the non-interacting automaton can be
explained by the fact that all of the particles can move simultaneously after
the first time step. \ Thus, the number of bits needed to represent the state
of the non-interacting automaton jumps after the first time step. \ With the
interacting automaton, by contrast, particles far from the coffee-cream border
cannot move until particles closer to the border have moved, so there is less
change in the automaton at each time step. \ Therefore, the estimated entropy
of this model is predictably more continuous throughout.

A visualization of the automaton's changing state over time is provided in
Figures \ref{interact} and \ref{noninteract}. \ This visualization is
generated by converting each cell's value to a grayscale color value; lighter
colors correspond to larger values. \ Visually, the fine-grained
representation of the state continues to grow more complicated with time,
while the coarse-grained representation first becomes first more and then less complicated.

\begin{figure}[h]\begin{center}%
\begin{tabular}
[c]{ccc}%
$\mathbf{t=0}$ & $\mathbf{t = 8 \times10^{6}}$ & $\mathbf{t = 2 \times10^{7}}%
$\\
\fbox{\includegraphics[height=75px, width=75px]{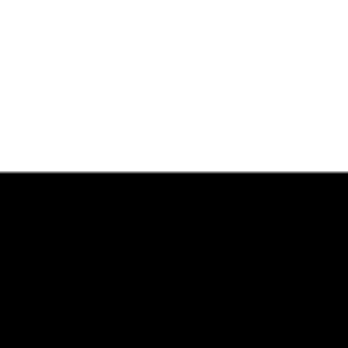}} &
\fbox{\includegraphics[height=75px, width=75px]{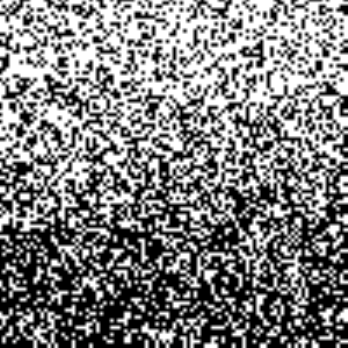}} &
\fbox{\includegraphics[height=75px, width=75px]{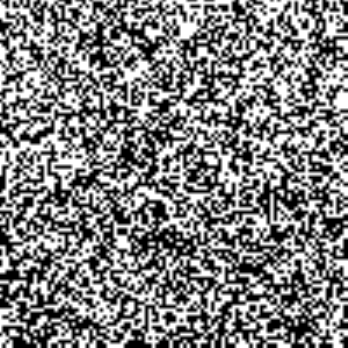}}\\
\fbox{\includegraphics[height=75px, width=75px]{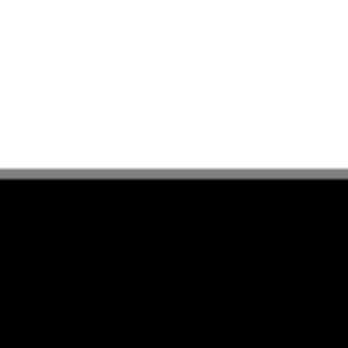}} &
\fbox{\includegraphics[height=75px, width=75px]{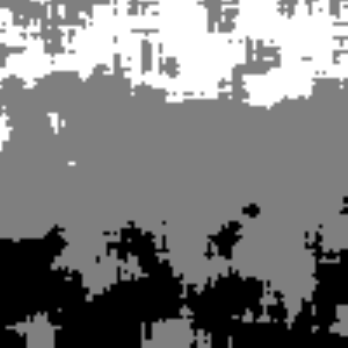}} &
\fbox{\includegraphics[height=75px, width=75px]{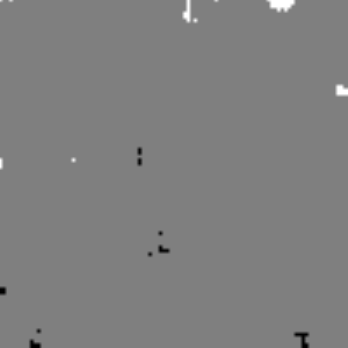}}\\
&  &
\end{tabular}
\caption{Visualization of the state of the interacting automaton of size $100$
over time. \ The top row of images is the fine-grained array, used to estimate
entropy. The bottom row is the coarse-grained array, used to estimate
complexity. \ From left to right, the images represent the automaton state at
the beginning of the simulation, at the complexity maximum, and at the end of
the simulation.}%
\label{interact}%
\end{center}\end{figure}

\begin{figure}[h]\begin{center}%
\begin{tabular}
[c]{ccc}%
$\mathbf{t=0}$ & $\mathbf{t = 5000}$ & $\mathbf{t = 10000}$\\
\fbox{\includegraphics[height=75px, width=75px]{fig3_00}} &
\fbox{\includegraphics[height=75px, width=75px]{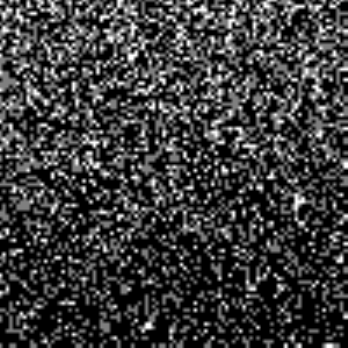}} &
\fbox{\includegraphics[height=75px, width=75px]{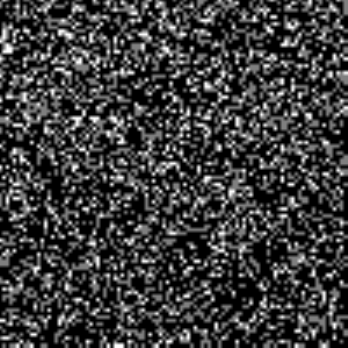}}\\
\fbox{\includegraphics[height=75px, width=75px]{fig3_10}} &
\fbox{\includegraphics[height=75px, width=75px]{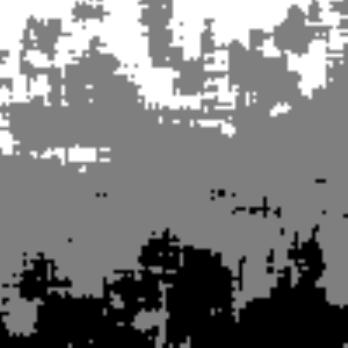}} &
\fbox{\includegraphics[height=75px, width=75px]{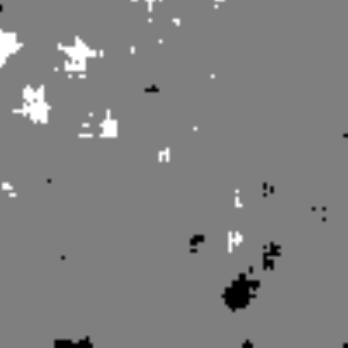}}\\
&  &
\end{tabular}
\caption{Visualization of the state of the non-interacting automaton of size
$100$ over time.}%
\label{noninteract}%
\end{center}\end{figure}

The \texttt{gzip} compression algorithm was used to generate the results in
Figure \ref{initialgraph}, and is used throughout when a general file
compression program is needed. \ The results achieved using the
coarse-graining metric are qualitatively similar when different compression
programs are used, as shown in Figure \ref{multicompress}.

\begin{figure}[h]\begin{center}
\fbox{\includegraphics[height=150px, width=200px]{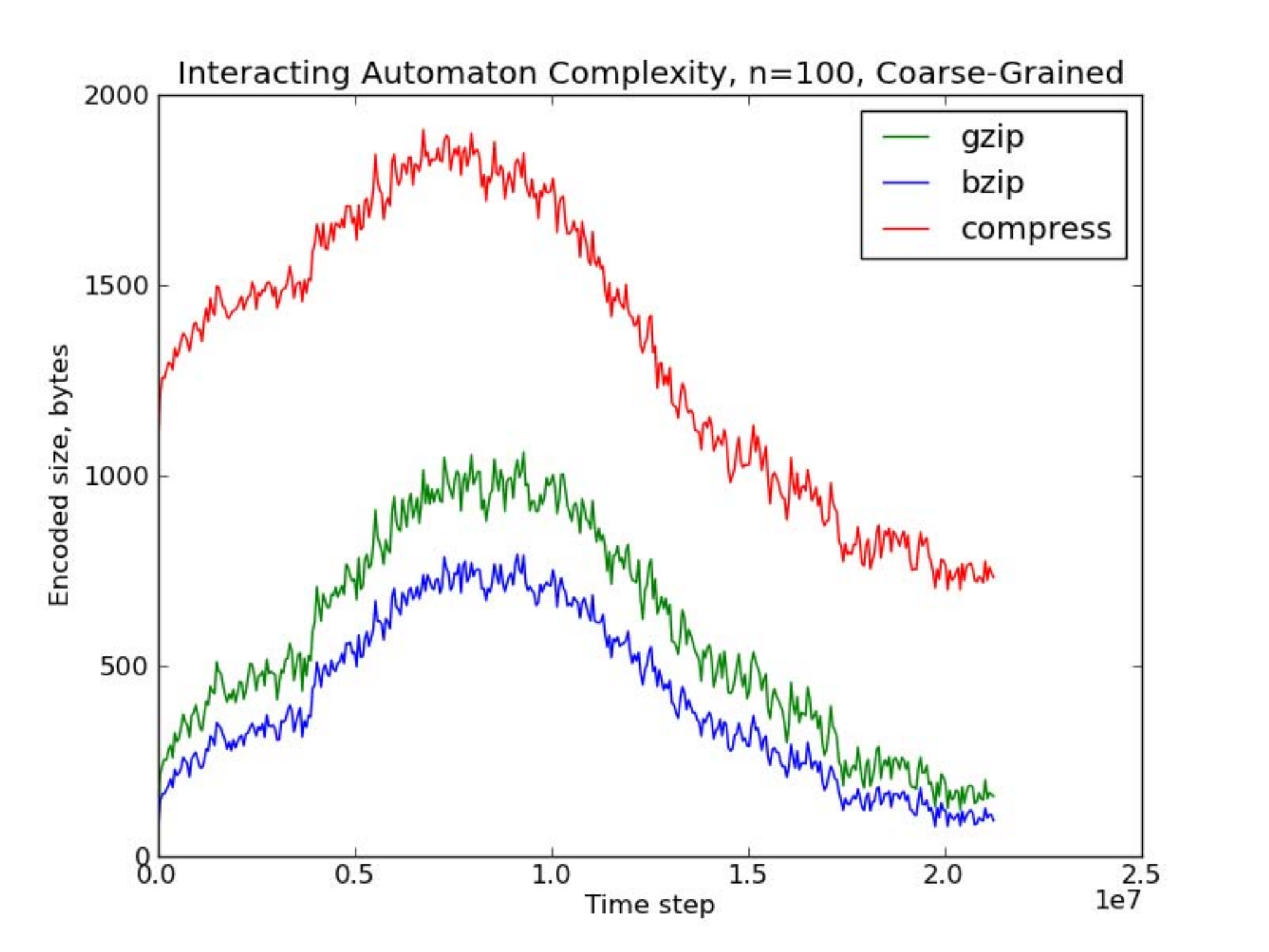}}%
\caption{Coarse-grained complexity estimates for a single simulation of the
interacting automaton, using multiple file compression programs.}%
\label{multicompress}%
\end{center}\end{figure}

Given these results, it is informative to examine how complexity varies with
$n$, the size of the automaton.

\begin{figure}[h]
\begin{center}
\fbox{\includegraphics[height=150px, width=200px]{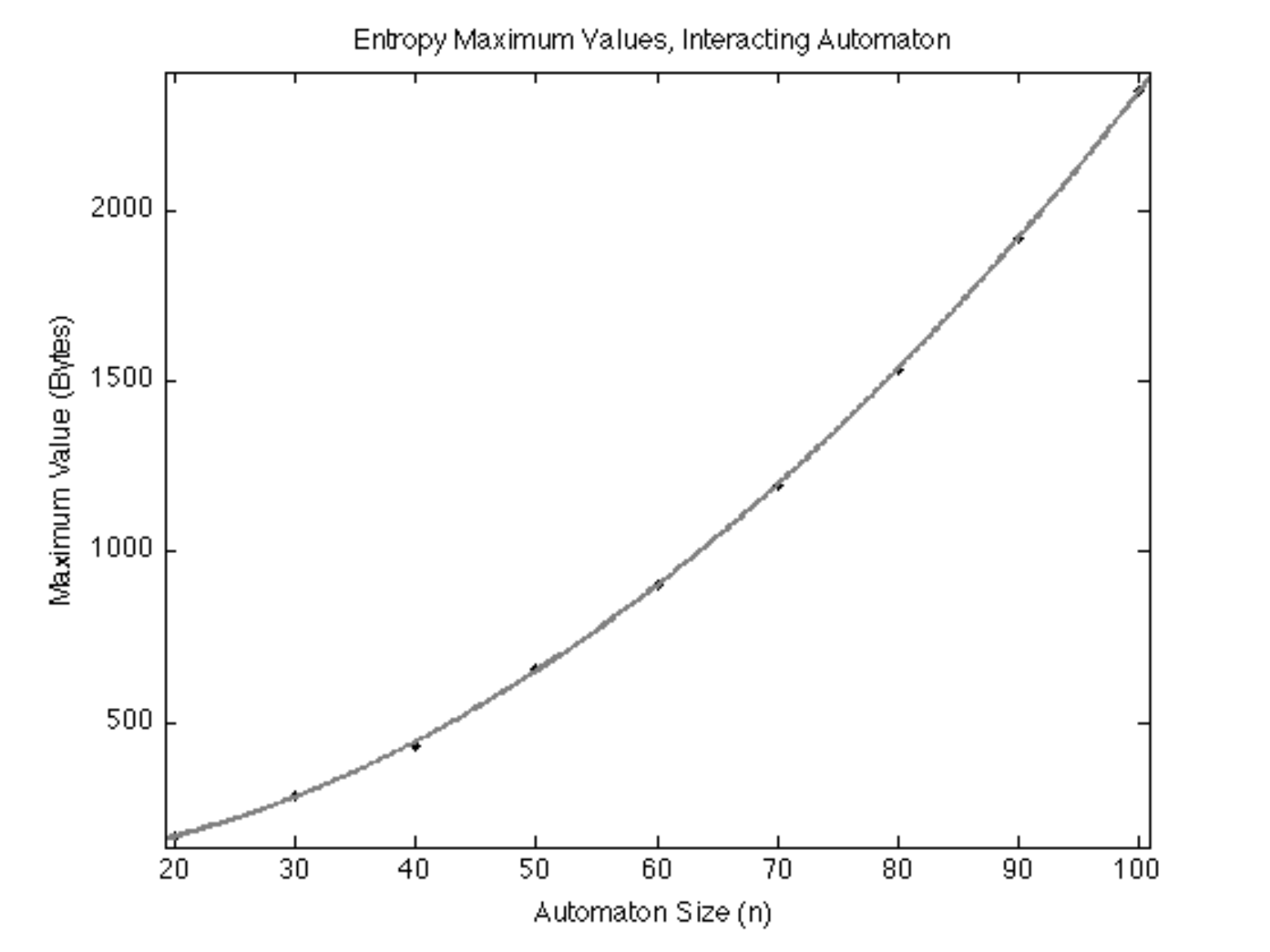}}
\fbox{\includegraphics[height=150px, width=200px]{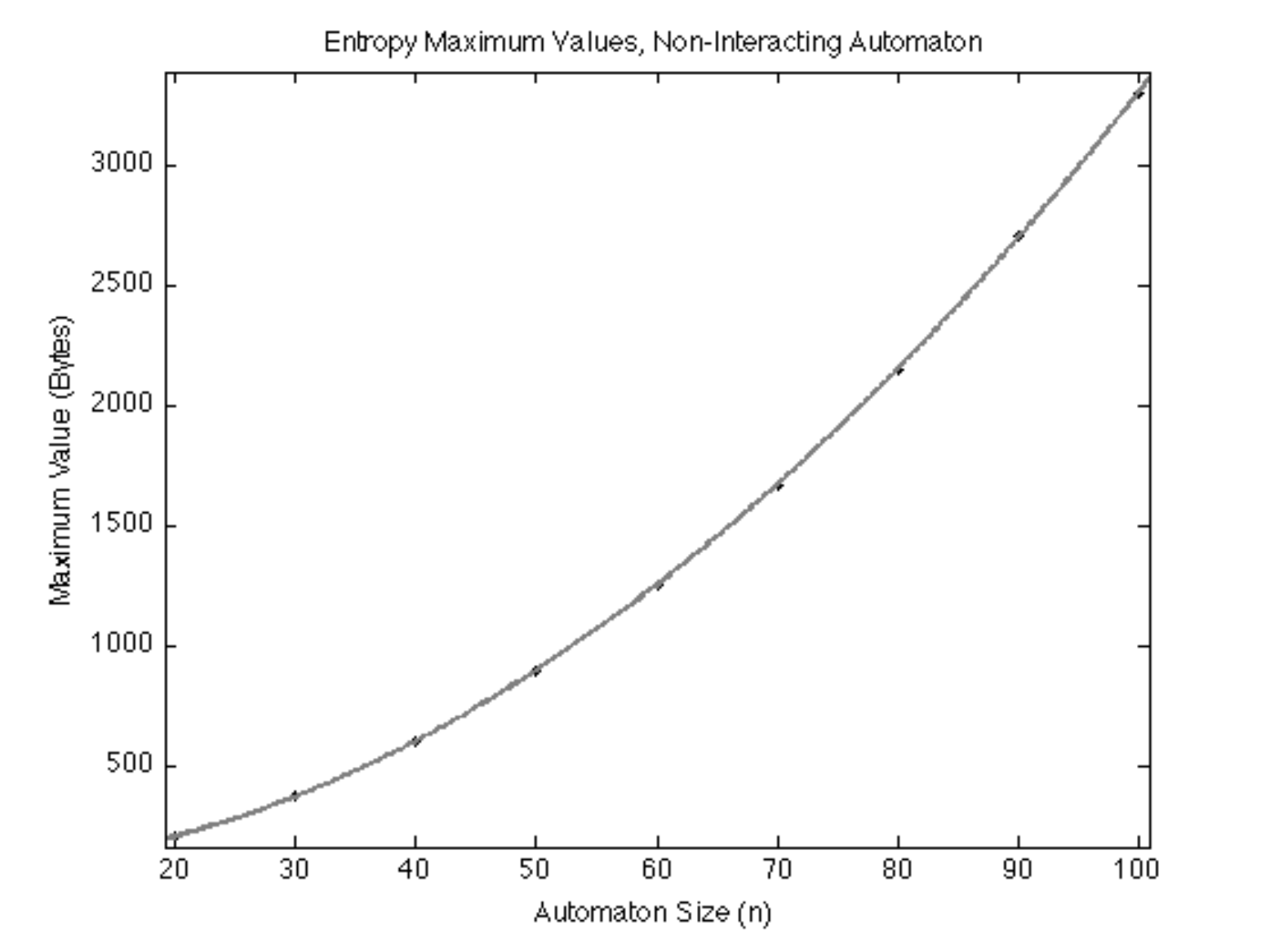}}%
\caption{Graphs of automaton size versus entropy maximum value. \ Quadratic
curve fits are shown, with $r^{2}$ values of $0.9999$ for both the interacting
and non-interacting automaton.}%
\label{nentropy}%
\end{center}\end{figure}\begin{figure}\begin{center}
\fbox{\includegraphics[height=150px, width=200px]{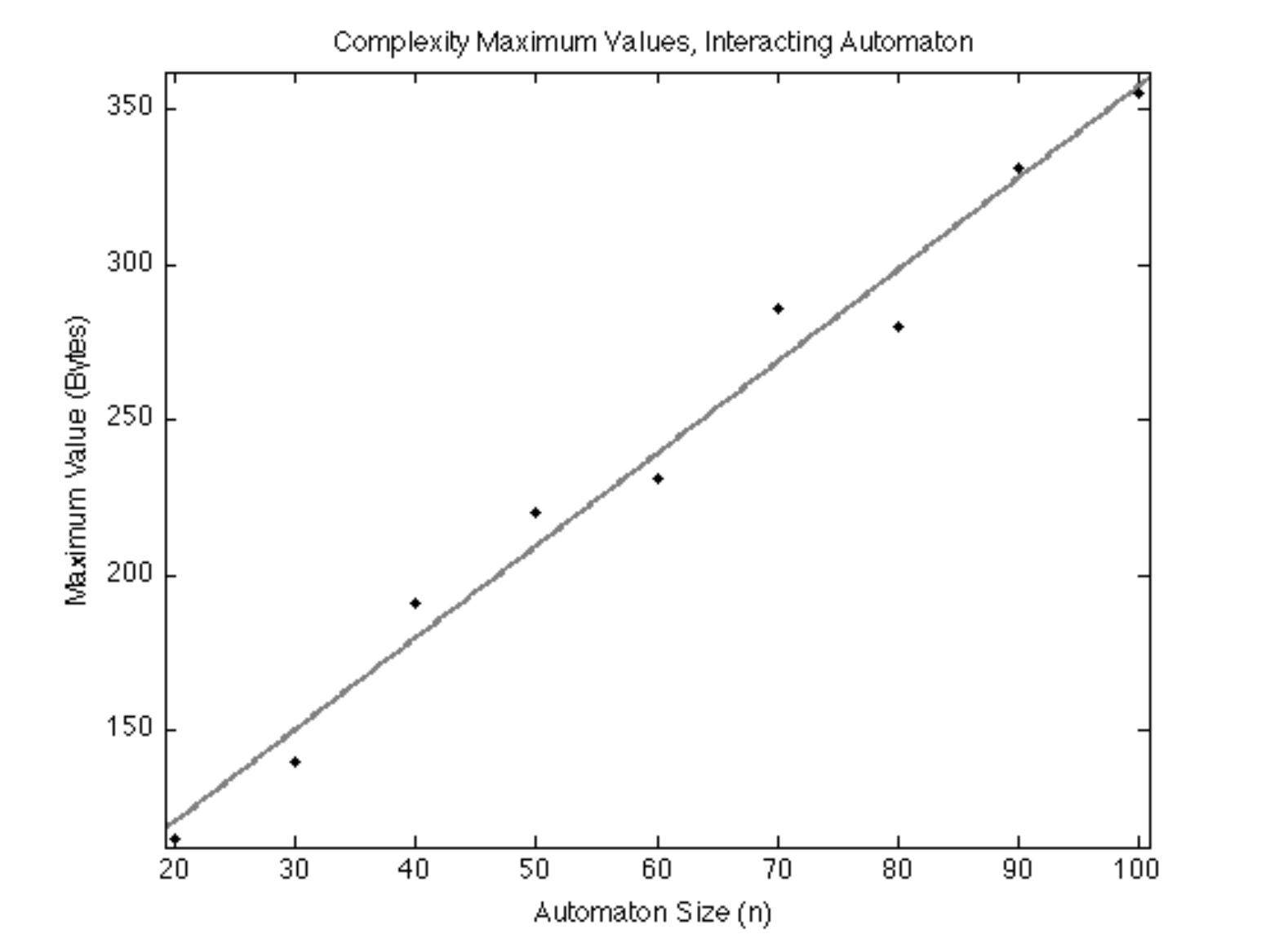}}
\fbox{\includegraphics[height=150px, width=200px]{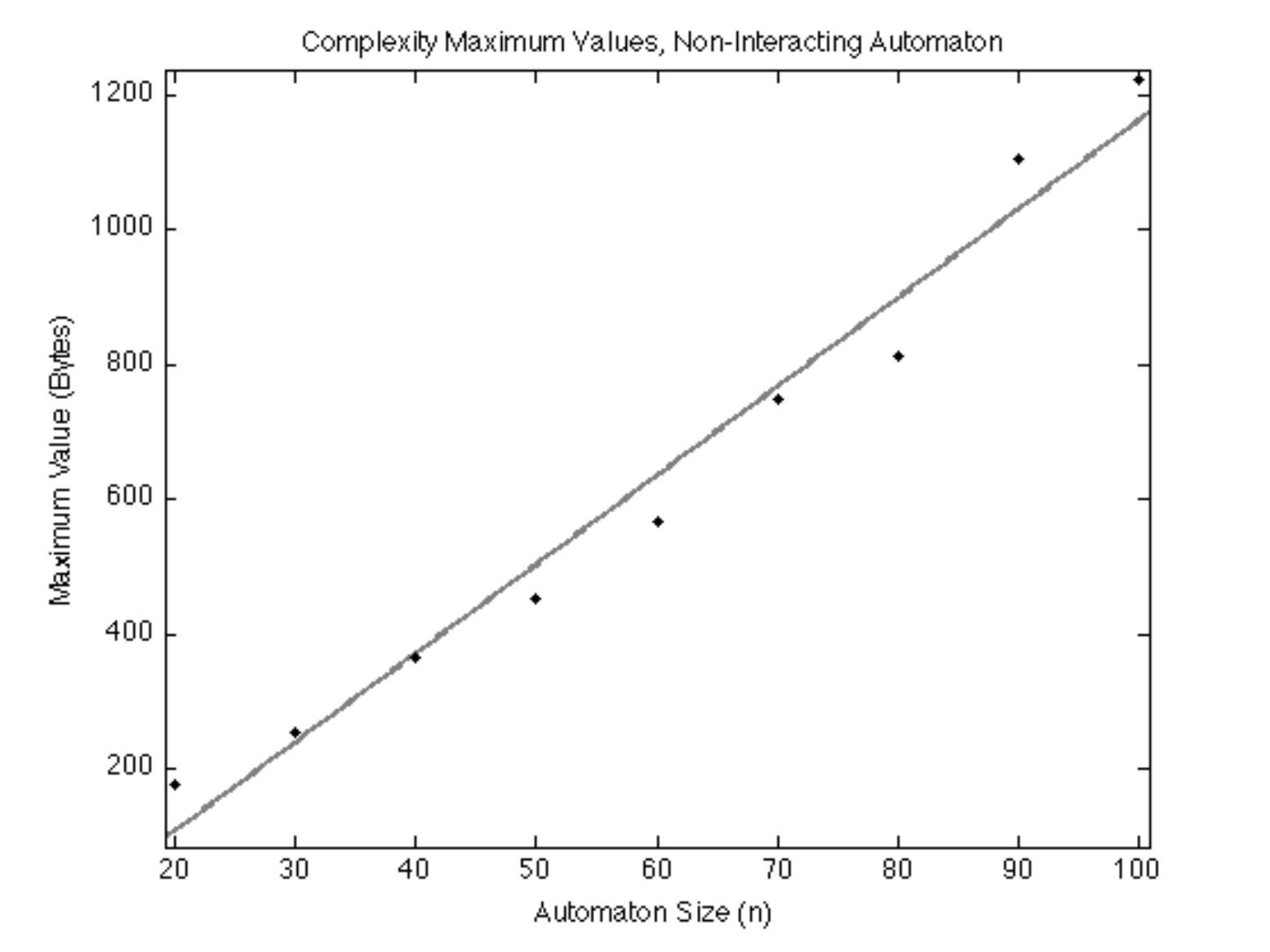}}%
\caption{Graphs of automaton size versus complexity maximum value. \ Linear
curve fits are shown, with $r^{2}$ values of $0.9798$ for the interacting
automaton and $0.9729$ for the non-interacting automaton.}%
\label{ncomplexity}%
\end{center}\end{figure}\begin{figure}\begin{center}
\fbox{\includegraphics[height=150px, width=200px]{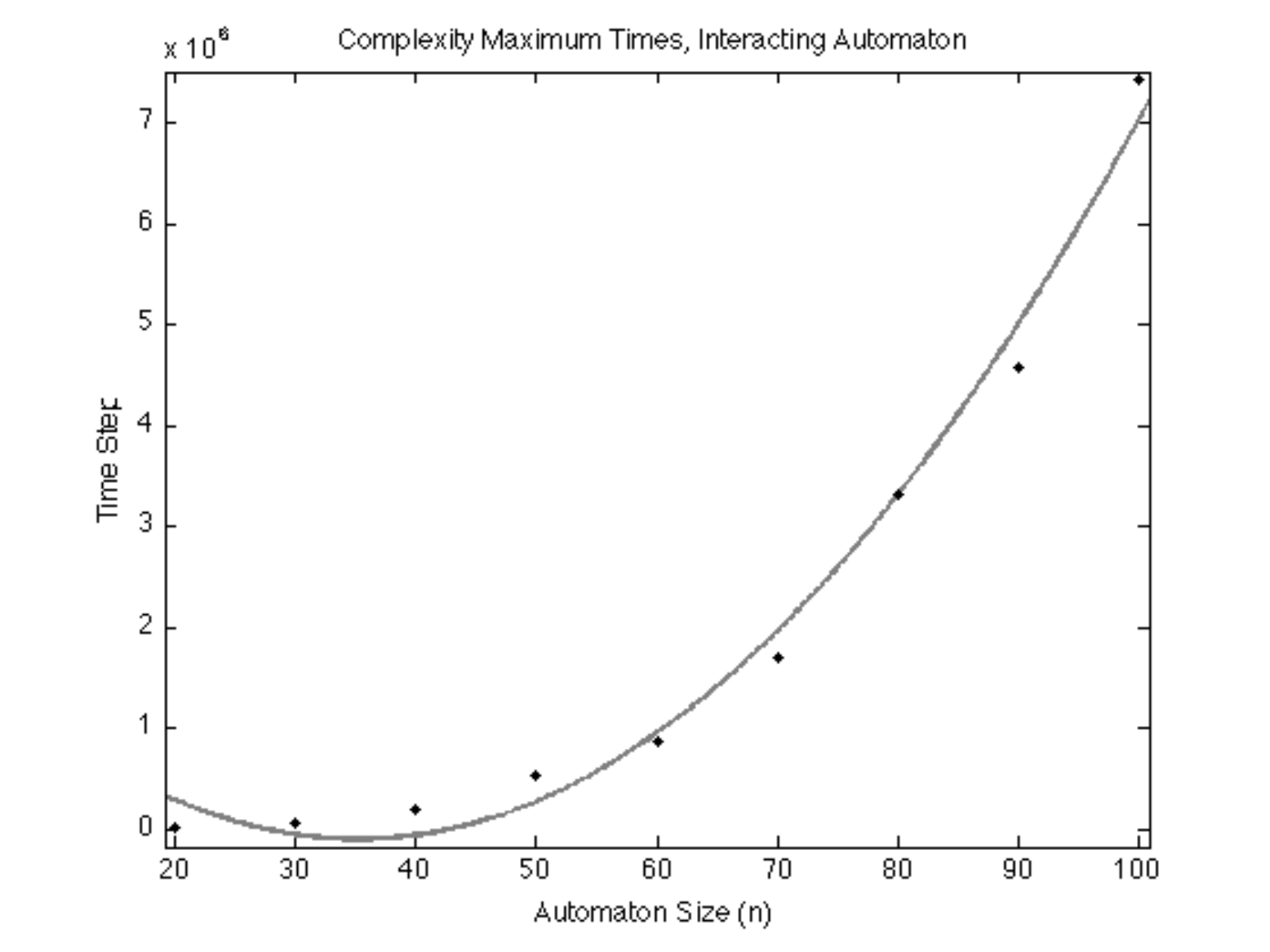}}
\fbox{\includegraphics[height=150px, width=200px]{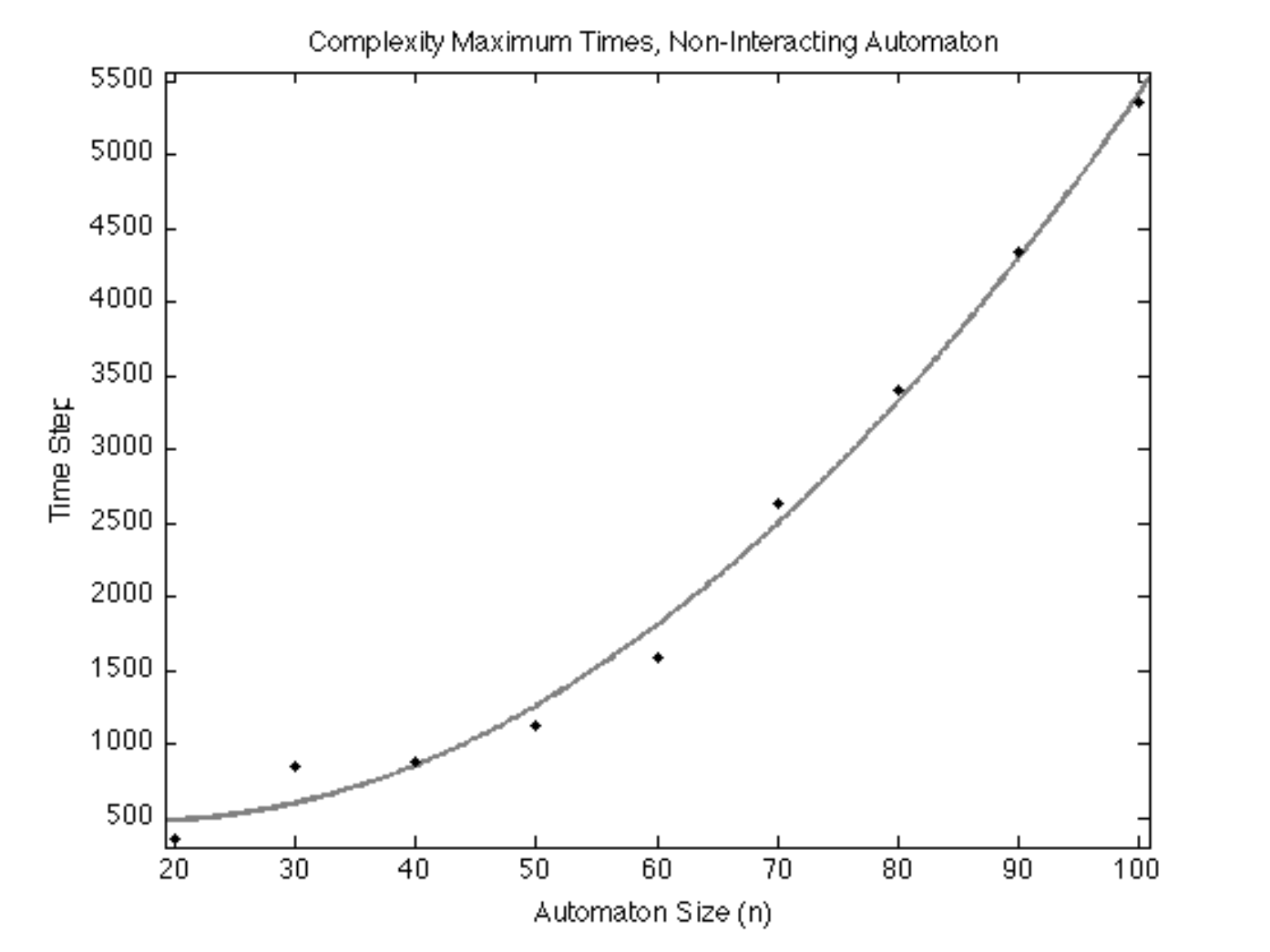}}%
\caption{Graphs of automaton size versus time to complexity maximum. Quadratic
curve fits are shown, with $r^{2}$ values of 0.9878 for the interacting
automaton and 0.9927 for the non-interacting automaton.}%
\label{timetomax}%
\end{center}\end{figure}

The well-fit quadratic curve for the maximum values of entropy (Figure
\ref{nentropy}) is expected. \ The maximum entropy of an automaton is
proportional to the number of particles in the automaton. \ This is because,
if the state of the automaton is completely random, then the compressed size
of the state is equal to the uncompressed size--the number of particles. \ As
the automaton size, $n$, increases, the number of particles increases to
$n^{2}$.

The maximum values of complexity appear to increase linearly as the automaton
size increases (Figure \ref{ncomplexity}). \ That is, maximum complexity is
proportional to the side length of the two-dimensional state array. \ This
result is expected, since the automaton begins in a state which is symmetric
along its vertical axis, and complexity presumably develops along a single
dimension of the automaton. \ The time that it takes for the automaton to
reach its complexity maximum appears to increase quadratically with the
automaton size, or proportionally to the number of particles in the automaton
(Figure \ref{timetomax}). \ This result is also expected, since the time for
$n^{2}$ particles to reach a particular configuration is proportional to
$n^{2}$.

\section{Adujsted Coarse-Graining Experiment\label{ACG}}

\subsection{Method\label{ACGMETHOD}}

Though the original coarse-graining approach produces the hypothesized
complexity pattern, the method of thresholding used in the previous
experiment---dividing the floating-point values into three buckets---has the
potential to introduce artificial complexity. \ Consider, for example, an
automaton state for which the coarse-grained array is a smooth gradient from
$0$ to $1$. \ By definition, there will be some row of the array which lies on
the border between two threshold values. \ Tiny fluctuations in the values of
the coarse-grained array may cause the cells in this row to fluctuate between
two threshold values. In such a case, the small fluctuations in this border
row would artificially increase the measured complexity of the coarse-grained
array. \ This case is illustrated in Figure \ref{coarsegraining}.

\begin{figure}\begin{center}
\fbox{\includegraphics[height=75px, width=75px]{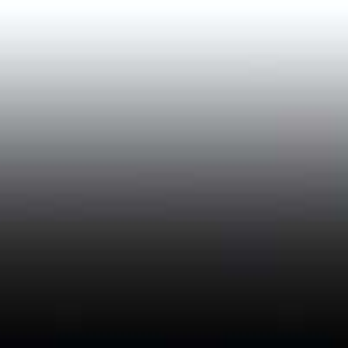}}
\fbox{\includegraphics[height=75px, width=75px]{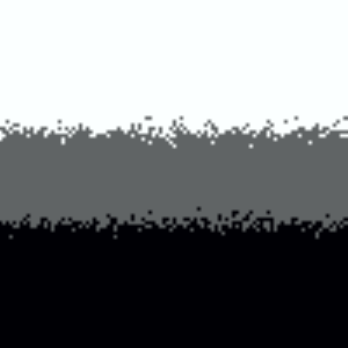}}%
\caption{A coarse-grained array consisting of a smooth gradient from $0$ to
$1$ is shown at left. \ At right is the same array after a small amount of
simulated noise has been added and the values have been thresholded.}%
\label{coarsegraining}%
\end{center}\end{figure}

We propose an adjustment to the coarse-graining algorithm that helps to
minimize these artifacts. \ First, we use a larger number of
thresholds---seven, in contrast to the three used in the original experiment.
\ Additionally, we allow each cell in the array to be optionally,
independently adjusted up or down by one threshold, in whatever manner
achieves the smallest possible file size for the coarse-grained array.

This adjustment helps to compensate for the thresholding artifacts--such
random fluctuations could be removed by adjusting the fluctuating pixels.
\ However, since each pixel can be adjusted independently, there are
$2^{n^{2}}$ possible ways to adjust a given coarse-grained array.

Because we cannot search through this exponential number of possible
adjustments to find the optimal one, we develop an approximation algorithm to
produce an adjustment that specifically targets pixels on the border between
two thresholds. \ Given the properties of the automaton---it begins with rows
of dark cells on top, and light cells on the bottom---it is likely that each
row of the coarse-grained array will contain similar values. \ Thus, we adjust
the coarse-grained array by using a majority algorithm. \ If a cell is within
one threshold value of the majority value in its row, it is adjusted to the
majority value.

The hope is that

\begin{enumerate}
\item[(1)] this adjustment will reduce artificial border complexity by
\textquotedblleft flattening\textquotedblright\ fluctuating border rows to a
single color,

\item[(2)] the adjustment will not eliminate actual complexity, since
complicated structures will create value differences in the coarse-grained
array that are large enough to span multiple thresholds.
\end{enumerate}

\subsection{Results and Analysis\label{ACGRESULTS}}

Results from simulation of the automaton using the adjusted coarse-graining
metric are shown in Figure \ref{bettergraph}. \ Visualizations of the
automaton state are shown in Figures \ref{size100} and
\ref{size100noninteract}.

\begin{figure}\begin{center}
\fbox{\includegraphics[height=150px, width=200px]{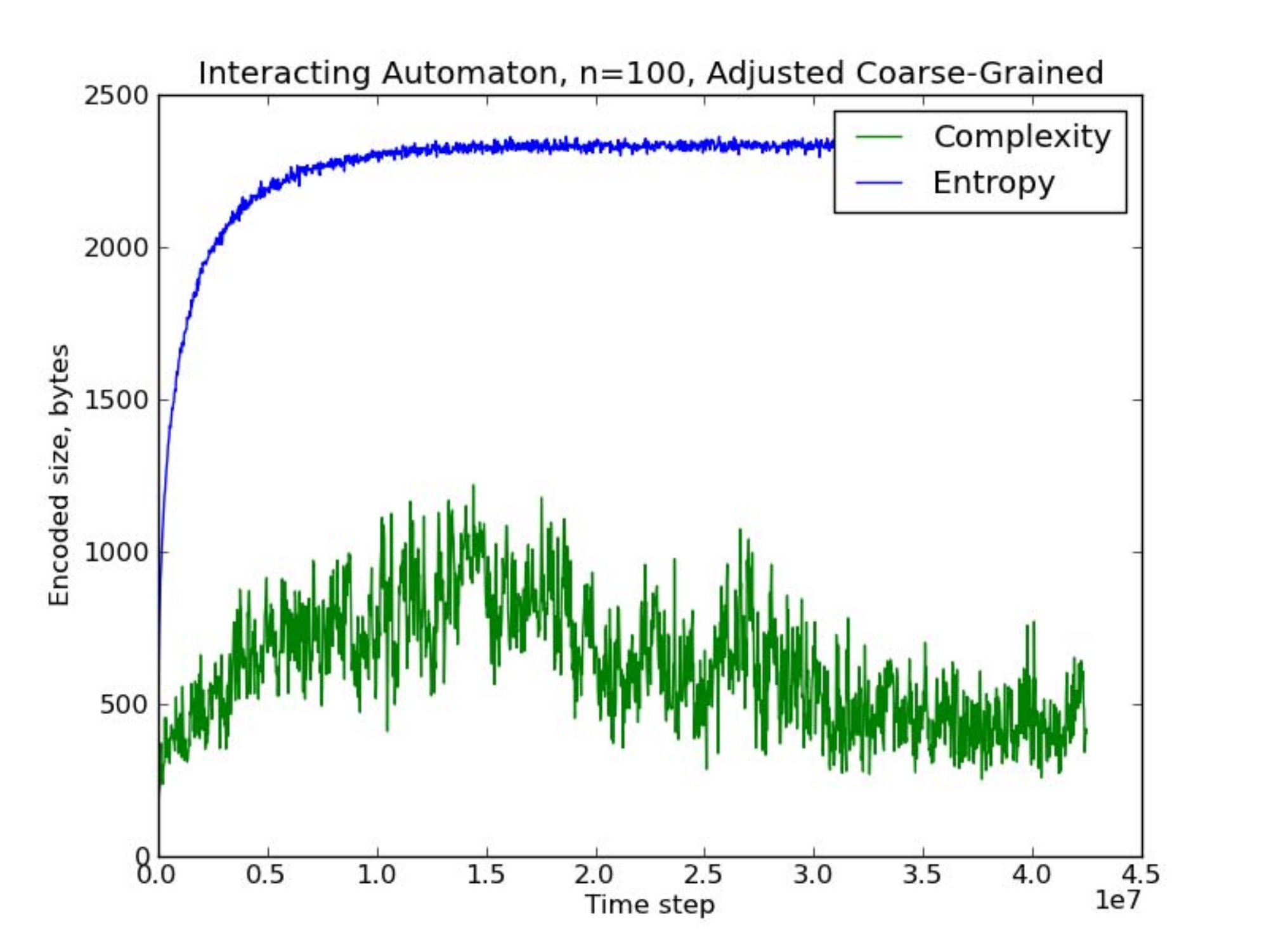}}
\fbox{\includegraphics[height=150px, width=200px]{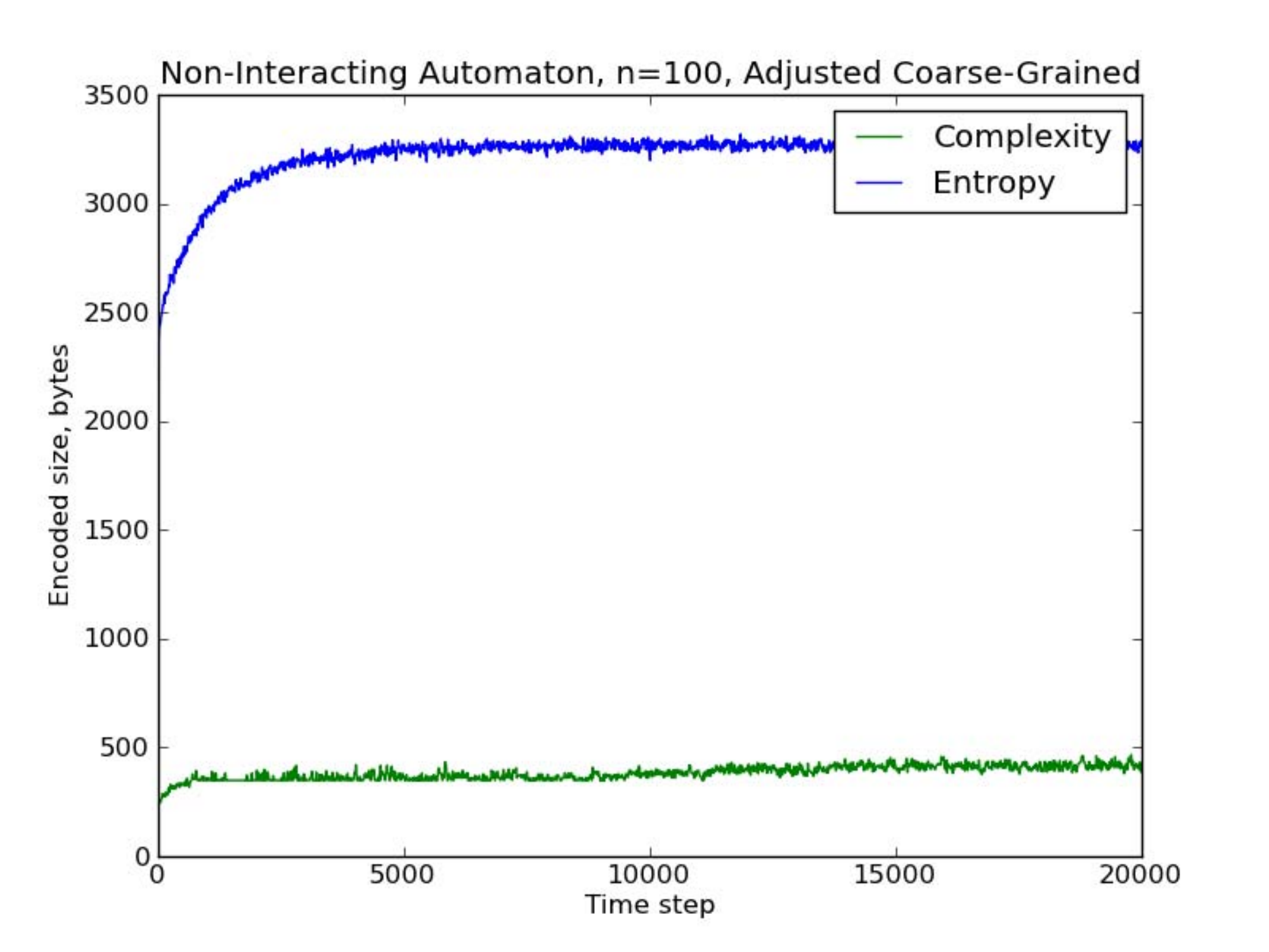}}%
\caption{The estimated entropy and complexity of an automaton using the
adjusted coarse-graining metric.}%
\label{bettergraph}%
\end{center}\end{figure}\begin{figure}\begin{center}%
\begin{tabular}
[c]{ccc}%
$\mathbf{t=0}$ & $\mathbf{t = 1.4 \times10^{7}}$ & $\mathbf{t = 4 \times
10^{7}}$\\
\fbox{\includegraphics[height=75px, width=75px]{fig3_00}} &
\fbox{\includegraphics[height=75px, width=75px]{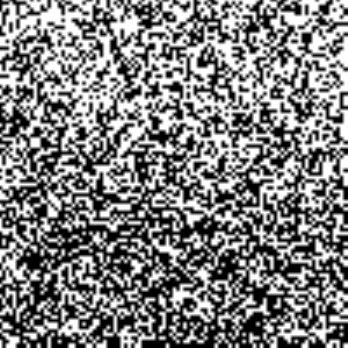}} &
\fbox{\includegraphics[height=75px, width=75px]{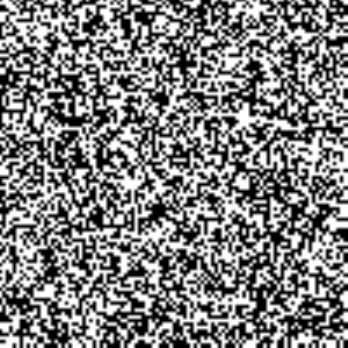}}\\
\fbox{\includegraphics[height=75px, width=75px]{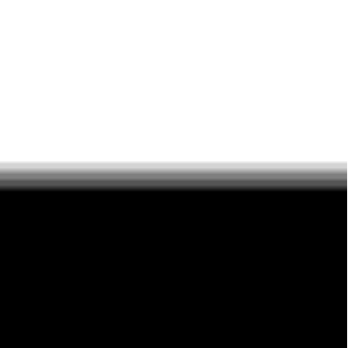}} &
\fbox{\includegraphics[height=75px, width=75px]{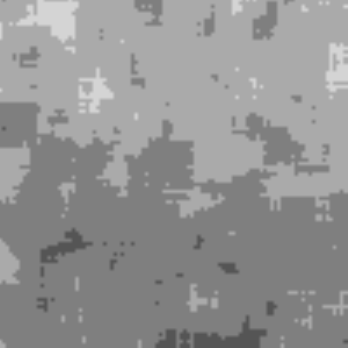}} &
\fbox{\includegraphics[height=75px, width=75px]{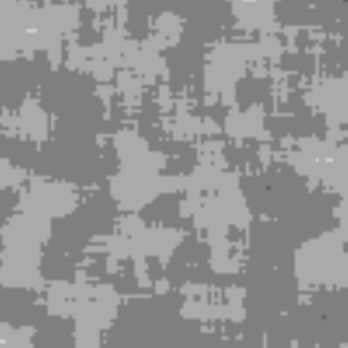}}\\
\fbox{\includegraphics[height=75px, width=75px]{fig11_10}} &
\fbox{\includegraphics[height=75px, width=75px]{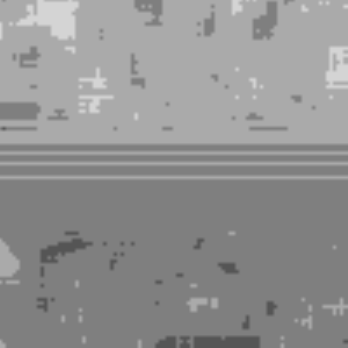}} &
\fbox{\includegraphics[height=75px, width=75px]{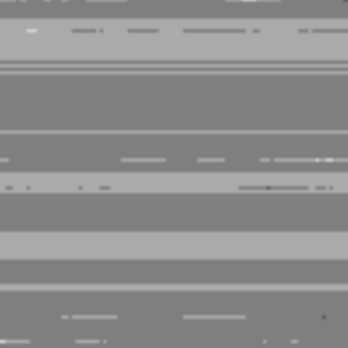}}\\
&  &
\end{tabular}
\caption{Visualization of the state of the interacting automaton of size $100$
over time. \ The rows of images represent the fine-grained state, the original
coarse-grained state, and the coarse-grained state after adjustment,
respectively.}%
\label{size100}%
\end{center}\end{figure}\begin{figure}\begin{center}%
\begin{tabular}
[c]{ccc}%
$\mathbf{t=0}$ & $\mathbf{t = 10000}$ & $\mathbf{t = 20000}$\\
\fbox{\includegraphics[height=75px, width=75px]{fig3_00}} &
\fbox{\includegraphics[height=75px, width=75px]{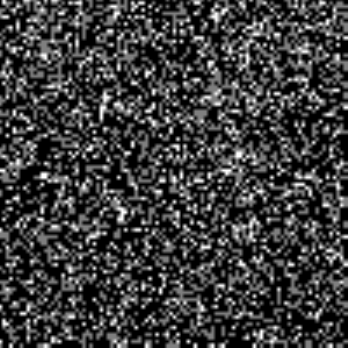}} &
\fbox{\includegraphics[height=75px, width=75px]{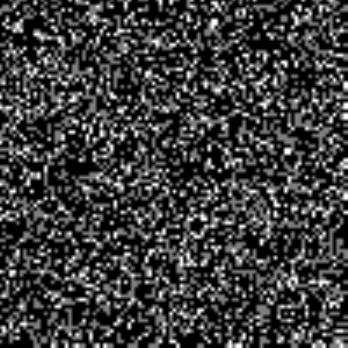}}\\
\fbox{\includegraphics[height=75px, width=75px]{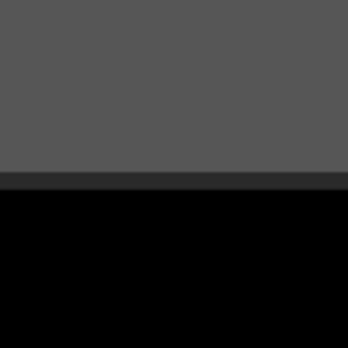}} &
\fbox{\includegraphics[height=75px, width=75px]{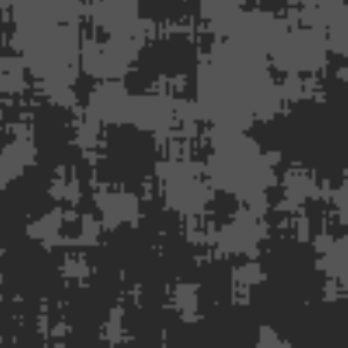}} &
\fbox{\includegraphics[height=75px, width=75px]{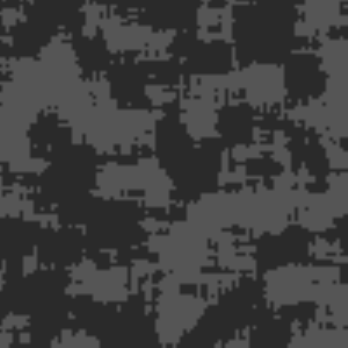}}\\
\fbox{\includegraphics[height=75px, width=75px]{fig12_10}} &
\fbox{\includegraphics[height=75px, width=75px]{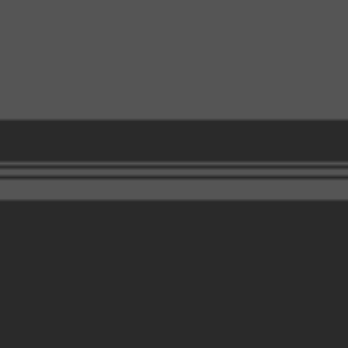}} &
\fbox{\includegraphics[height=75px, width=75px]{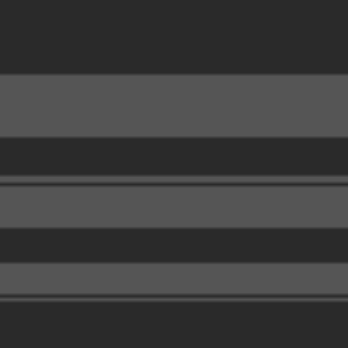}}\\
&  &
\end{tabular}
\caption{Visualization of the state of the non-interacting automaton of size
$100$ over time. Note that the coarse-grained images are darker than for the
previous coarse-graining metric, because a larger number of thresholds were
used.}%
\label{size100noninteract}%
\end{center}\end{figure}

While this metric is somewhat noisier than the original coarse-graining
method, it results in a similarly-shaped complexity curve for the interacting
automaton. \ For the non-interacting automaton, however, the complexity curve
is flattened to a lower value.

This result for the non-interacting automaton is actually borne out by
theoretical predictions. \ The basic story is as follows; for details, see
Appendix \ref{ANALYSIS}. \ If we consider the automaton state to
\textquotedblleft wrap around\textquotedblright\ from right to left, then by
symmetry, the expected number of cream particles at a particular location in
the automaton depends solely on the vertical position of that location. \ The
expectations of all cells in a particular row will be the same, allowing the
two-dimensional automaton state to be specified using a single dimension.
\ Modeling each particle of cream as taking a random walk from its initial
position, it is possible to calculate the expected number of particles at a
given position as a function of time. \ Further, Chernoff bounds can be used
to demonstrate that the actual number of particles in each grain of the
coarse-grained state is likely to be close to the expectation, provided that
the grain size is large enough. \ Since it is possible to specify the expected
distribution of particles in the non-interacting automaton at all times using
such a function, the complexity of the non-interacting automaton state is
always low.

We believe thresholding artifacts caused the apparent increase in complexity
for the non-interacting automaton when regular coarse-graining was used. \ Our
adjustment removes all of this estimated complexity from the non-interacting
automaton, but preserves it in the interacting automaton. \ This evidence
suggests that the interacting automaton model may actually have intermediate
states of high complexity, even if the non-interacting model never becomes complex.

\section{Conclusions and Further Work\label{CONC}}

Of the metrics considered in this project, the coarse-graining approaches such
as apparent complexity provide the most effective estimate of complexity that
produces results which mirror human intuition. \ However, this metric suffers
from the disadvantage that it is based on human intuition and perceptions of
complexity. \ Ideally, a complexity metric would be found which produces
similar results without relying on such assumptions. \ The OSCR approach seems
promising for its independence from these assumptions and for its theoretical
foundations. \ It is possible that a different implementation of this
algorithm could produce better results than the one we used for this project.

It would also be worthwhile to investigate other complexity metrics, beyond
those already explored in this paper. \ Shalizi et al.\ \cite{ssh} propose a
metric based on the concept of light cones. \ They define $C\left(  x\right)
$, the complexity of a point $x$ in the spacetime history, as the mutual
information between descriptions of its past and future light cones. \ Letting
$P\left(  x\right)  $ be the past light cone and $F\left(  x\right)  $ the
future light cone, $C\left(  x\right)  =H\left(  P\left(  x\right)  \right)
+H\left(  F\left(  x\right)  \right)  -H\left(  P\left(  x\right)  ,F\left(
x\right)  \right)  $. \ This metric is of particular interest because it
avoids the problem of artifacts created by coarse-graining; it can also be
approximated in a way that avoids the use of \texttt{gzip}. \ Running
experiments with the automaton using the light cone metric, and comparing the
results to those generated using coarse-graining, could provide more
information about both metrics.

Ultimately, numerical simulation is of limited use in reasoning about the
problem of complexity. \ Approximation algorithms can provide only an upper
bound, not a lower bound, on Kolmogorov complexity and sophistication. \ To
show that a system really does become complex at intermediate points in time,
it is necessary to find a lower bound for the system's complexity. \ Future
theoretical work could help provide such a lower bound, and could also
generate further insight into the origins of complexity in closed systems.

\section{Acknowledgments}

We thank Alex Arkhipov, Charles Bennett, Ian Durham, Dietrich Leibfried, Aldo
Pacchiano, and Luca Trevisan for helpful discussions.

\section{Appendix: The Non-Interacting Case\label{ANALYSIS}}

Let's consider the non-interacting coffee automaton on an $n\times n$ grid
with periodic boundary conditions. \ At each time step, each cream particle
moves to one of the $4$ neighboring pixels uniformly at random. \ Let
$a_{t}\left(  x,y\right)  $ be the number of cream particles at point $\left(
x,y\right)  $ after $t$ steps.

\begin{claim}
\label{eclaim}For all $x,y,t$, we have $\operatorname*{E}\left[  a_{t}\left(
x,y\right)  \right]  \leq1$.
\end{claim}

\begin{proof}
By induction on $t$. \ If $t=0$, then $a_{0}\left(  x,y\right)  \in\left\{
0,1\right\}  $. \ Furthermore, by linearity of expectation,%
\[
\operatorname*{E}\left[  a_{t+1}\left(  x,y\right)  \right]  =\frac
{\operatorname*{E}\left[  a_{t}\left(  x-1,y\right)  \right]
+\operatorname*{E}\left[  a_{t}\left(  x+1,y\right)  \right]
+\operatorname*{E}\left[  a_{t}\left(  x,y-1\right)  \right]
+\operatorname*{E}\left[  a_{t}\left(  x,y+1\right)  \right]  }{4}.
\]
\end{proof}

Now let $B$ be an $L\times L$\ square of pixels, located anywhere on the
$n\times n$\ grid. \ Let $a_{t}\left(  B\right)  $\ be the number of cream
particles in $B$ after $t$ steps. \ Clearly%
\begin{equation}
a_{t}\left(  B\right)  =\sum_{\left(  x,y\right)  \in B}a_{t}\left(
x,y\right)  .
\end{equation}
So it follows from Claim \ref{eclaim}\ that $\operatorname*{E}\left[
a_{t}\left(  B\right)  \right]  \leq L^{2}$.

Fix some constant $G$, say $10$. \ Then call $B$ \textquotedblleft
bad\textquotedblright\ if $a_{t}\left(  B\right)  $\ differs from
$\operatorname*{E}\left[  a_{t}\left(  B\right)  \right]  $ by more than
$L^{2}/G$. \ Suppose that at some time step $t$, no $B$ is bad. \ Also,
suppose we form a coarse-grained image by coloring each $B$ one of $G$ shades
of gray, depending on the value of%
\begin{equation}
\left\lfloor \frac{a_{t}\left(  B\right)  G}{L^{2}}\right\rfloor
\end{equation}
(or we color $B$ white if $a_{t}\left(  B\right)  >L^{2}$). \ Then it's clear
that the resulting image will be correctable, by adjusting each color by
$\pm1$, to one where all the $B$'s within the same row are assigned the same
color---and furthermore, that color is simply%
\begin{equation}
\left\lfloor \frac{\operatorname*{E}\left[  a_{t}\left(  B\right)  \right]
G}{L^{2}}\right\rfloor .
\end{equation}
If this happens, though, then the Kolmogorov complexity of the coarse-grained
image can be at most $\log_{2}\left(  n\right)  +\log_{2}\left(  t\right)
+O\left(  1\right)  $. \ For once we've specified $n$ and $t$, we can simply
\textit{calculate} the expected color for each $B$, and no color ever deviates
from its expectation.

So our task reduces to upper-bounding the probability that $B$ is bad. \ By a
Chernoff bound, since $a_{t}\left(  B\right)  $\ is just a sum of independent,
$0/1$ random variables,%
\begin{equation}
\Pr\left[  \left\vert a_{t}\left(  B\right)  -\operatorname*{E}\left[
a_{t}\left(  B\right)  \right]  \right\vert >\delta\operatorname*{E}\left[
a_{t}\left(  B\right)  \right]  \right]  <2\exp\left(  -\frac
{\operatorname*{E}\left[  a_{t}\left(  B\right)  \right]  \delta^{2}}%
{3}\right)  .
\end{equation}
Plugging in $L^{2}/G=\delta\operatorname*{E}\left[  a_{t}\left(  B\right)
\right]  $, we get%
\begin{equation}
\Pr\left[  \left\vert a_{t}\left(  B\right)  -\operatorname*{E}\left[
a_{t}\left(  B\right)  \right]  \right\vert >\frac{L^{2}}{G}\right]
<2\exp\left(  -\frac{L^{4}}{3G^{2}\operatorname*{E}\left[  a_{t}\left(
B\right)  \right]  }\right)  .
\end{equation}
Since $\operatorname*{E}\left[  a_{t}\left(  B\right)  \right]  \leq L^{2}$
from above, this in turn is at most%
\begin{equation}
2\exp\left(  -\frac{L^{2}}{3G^{2}}\right)  .
\end{equation}
Now, provided we choose a coarse-grain size%
\begin{equation}
L\gg G\sqrt{3\ln\left(  2n^{2}\right)  }=\Theta\left(  G\sqrt{\log n}\right)
,
\end{equation}
the above will be much less than $1/n^{2}$. \ In that case, it follows by the
union bound that, at each time step $t$, with high probability \textit{none}
of the $L\times L$\ squares $B$ are bad (since there at most $n^{2}$ such
squares). \ This is what we wanted to show.

\bibliographystyle{plain}
\bibliography{coffee}

\end{document}